\documentclass[lettersize,journal]{IEEEtran}

\usepackage{amsmath,graphicx,epstopdf,amsthm}%
\usepackage[normalem]{ulem}
 \usepackage[caption=false,font=normalsize,labelfont=sf,textfont=sf]{subfig}
\usepackage{hyperref}
\usepackage{epstopdf} 
\usepackage{amsmath,bm}    
\usepackage{verbatim} 
\usepackage{makecell}
\usepackage{array}

\usepackage[utf8]{inputenc}     
\usepackage[english]{babel}   
\usepackage{caption}    
\usepackage{enumitem} 
\usepackage{amssymb}
\usepackage{xcolor}
\usepackage{cite}

\usepackage{booktabs,caption,siunitx}
\DeclareMathOperator*{\argmin}{arg\,min}

\usepackage[ruled,vlined,linesnumbered,noend]{algorithm2e}
\DeclareCaptionLabelFormat{lc}{\MakeLowercase{#1}~#2}
\SetKwFor{NodeFor}{each node}{in parallel:}{}

\newtheorem{fact}{Fact}
\newtheorem{assumption}{Assumption}
\newtheorem{corollary}{Corollary}
\newtheorem{lemma}{Lemma}
\newtheorem{definition}{Definition}
\newtheorem{proposition}{Proposition}
\newtheorem{theorem}{Theorem}

\allowdisplaybreaks

    {\endgroup}

\AtBeginDocument{%
  \providecommand\BibTeX{{%
    \normalfont B\kern-0.5em{\scshape i\kern-0.25em b}\kern-0.8em\TeX}}}

\begin{document}
\title{Differentially-Private Multi-Tier Federated Learning: A Formal Analysis and Evaluation}

\author{
 ~Frank Po-Chen Lin\IEEEauthorrefmark{1},~\IEEEmembership{Student Member,~IEEE},~Evan Chen\IEEEauthorrefmark{1},~\IEEEmembership{Student Member,~IEEE},
\\ Dong-Jun Han,~\IEEEmembership{Member,~IEEE}, ~Christopher G. Brinton,~\IEEEmembership{Senior Member,~IEEE}

\thanks{Evan Chen, Frank Po-Chen Lin and Christopher G. Brinton are with the Elmore Family School of Electrical and Computer Engineering, Purdue University, West Lafayette, IN, 47907, USA (e-mail: chen4388@purdue.edu, frank555076@gmail.com, cgb@purdue.edu).}
\thanks{Dong-Jun Han is with the Department of Computer Science and Engineering, Yonsei University, Seoul 03722, South Korea (e-mail: djh@yonsei.ac.kr).}
\thanks{An abridged version of this paper is under review in the 2025 IEEE International Conference on Communications (ICC).}
\thanks{This work was supported by the National Science Foundation (NSF) under grants CNS-2146171 and CPS-2313109, and by the Office of Naval Research (ONR) under grant N000142212305.}
}

\maketitle
\setlength{\footnotesep}{\baselineskip}
\def\thefootnote{\IEEEauthorrefmark{1}}\footnotetext{These authors contributed equally to this work.}\def\thefootnote{\arabic{footnote}}

\begin{abstract}
While federated learning (FL) eliminates the transmission of raw data over a network, it is still vulnerable to privacy breaches from the communicated model parameters. Differential privacy (DP) is often employed to address such issues. However, the impact of DP on FL in multi-tier networks -- where hierarchical aggregations couple noise injection decisions at different tiers, and trust models are heterogeneous across subnetworks -- is not well understood. To fill this gap, we develop \underline{M}ulti-Tier \underline{F}ederated Learning with \underline{M}ulti-Tier \underline{D}ifferential \underline{P}rivacy ({\tt M$^2$FDP}), a DP-enhanced FL methodology for jointly optimizing privacy and performance over such networks. One of the key principles of {\tt M$^2$FDP} is to adapt DP noise injection across the established edge/fog computing hierarchy (e.g., edge devices, intermediate nodes, and other tiers up to cloud servers) according to the trust models in different subnetworks. We conduct a comprehensive analysis of the convergence behavior of {\tt M$^2$FDP} under non-convex problem settings, revealing conditions on parameter tuning under which the training process converges sublinearly to a finite stationarity gap that depends on the network hierarchy, trust model, and target privacy level. 
We show how these relationships can be employed to develop an adaptive control algorithm for {\tt M$^2$FDP} that tunes properties of local model training to minimize energy, latency, and the stationarity gap while meeting desired convergence and privacy criterion.
Subsequent numerical evaluations demonstrate that {\tt M$^2$FDP} obtains substantial improvements in these metrics over baselines for different privacy budgets and system configurations.
\end{abstract}     

\begin{IEEEkeywords}
Federated Learning, Edge Intelligence, Differential Privacy, Hierarchical Networks
\end{IEEEkeywords}

 

\maketitle
\vspace{-0.10in}
\section{Introduction}
The concept of privacy has significantly evolved in the digital age, particularly with regards to data collection, sharing, and utilization in machine learning (ML) algorithms~\cite{Mohammad2019,Emiliano2021}. The ability to extract knowledge from massive datasets is a double-edged sword; while empowering ML algorithms, it simultaneously exposes individuals' sensitive information. Therefore, it is crucial to develop ML methods that respect user privacy and conform to data protection standards~\cite{Farokhi2021priv,Chu2022MitigatingBI}.

To this end, federated learning (FL) has emerged as an attractive paradigm for distributing ML over networks, as it allows for model updates to occur directly on the edge devices where the data originates~\cite{mcmahan2017communication,wang2019adaptive,haddadpour2019convergence,lin2021timescale,ju2024can}. Information transmitted over the network is in the form of locally trained models (rather than raw data), which are periodically aggregated at a central server. 
Nonetheless, FL is still susceptible to privacy threats: it has been shown that adversaries with access to model updates can reverse engineer attributes of device-side data~\cite{Wei2020,Zhu2019,Wang2019TM}. This has motivated different threads of investigation on privacy preservation within the FL framework. One promising approach has been the introduction of differential privacy (DP) into FL~\cite{Poor2020LDP,Shen2022imp,Zhao2021LDP,Shi2021HDP,chandrasekaran2024hierarchical,Wang2022DP,Sun2022DP,Zhang2020Mobi,Liu2023mobi}. DP injects calibrated noise into the model parameters to prevent the leakage of individual-level information, creating a privacy-utility tradeoff for FL.

DP in FL has been mostly studied for conventional FL architectures, where client devices are directly connected to the main server that conducts global model aggregations, i.e., in a star topology. However, in practice, direct communication with a central cloud server from the wireless edge can become prohibitively expensive, especially for large ML models and a large number of participating clients spread over a large geographic region. Much research has been devoted to enhancing the communication efficiency of this global aggregation step, e.g., \cite{wang2022federated,amiri2020federated,li2021talk,karimireddy2020scaffold,zhang2021low}.
In particular, recent studies have explored the potential of \textit{decentralizing} the star topology assumption in FL to account for more realistic networks comprised of edge devices, fog nodes, and the cloud \cite{wang2024device,mishchenko2022proxskip,lin2021timescale,deng2024communication,liu2022hierarchical}. By introducing localized communication strategies within the intermediate edge/fog layers, these approaches aim to minimize the overhead at the central cloud server. Much of this work has focused on \textit{Hierarchical FL} (HFL) \cite{hosseinalipour2020federated,nguyen2022fedfog,lin2021timescale,chen2024taming,chandrasekaran2024hierarchical}, considering a two-layer device-edge-cloud network architecture; more generally, \textit{Multi-Tier FL} (MFL) \cite{hosseinalipour2022multi,fang2024hierarchical} considers the distribution of FL tasks across computing architectures with an arbitrary number of intermediate fog layers.

Despite the practical significance of MFL, there is not yet a concrete understanding of how to apply and optimize DP within this framework. 
Indeed, the intermediate tiers between edge and cloud offer additional flexibility into where and how DP noise injection occurs, but challenge our understanding of how DP impacts ML training performance. Specifically, the hierarchical and heterogeneous structure introduce the following key challenges:
\begin{itemize}[leftmargin=8mm]
\item \textit{Noise injection coupling across different tiers:} In multi-tier networks, the injection of DP noise at any particular node will propagate to higher tiers through the HFL aggregation structure. Excessive noise introduced at one tier can accumulate and significantly degrade the ML training performance. The coupled impact of noise injection decisions thus requires careful consideration. 

\item \textit{Heterogeneous trust models across subnetworks:} Each model aggregation path from edge to cloud will pass through a different set of nodes, and thus a different set of trust models. This results in heterogeneous privacy constraints across different subnetworks, calling for a tailored strategy of noise injection to carefully balance ML training performance and DP guarantees.
\end{itemize}
In this work, we aim to address these challenges by conducting a comprehensive examination of the privacy-utility tradeoff associated with DP injection across MFL systems. Specifically, we seek to answer the following research questions:
\begin{enumerate}[leftmargin=8mm]
\item \textit{What is the coupled effect between MFL system configuration and multi-tier DP noise injection on ML performance?}

\item \textit{How can we adapt MFL and DP parameters to jointly optimize model performance, privacy preservation, and resource utilization?}
\vspace{-2mm}
\end{enumerate}

\subsection{Related Work}
DP has been widely studied as a means of preventing extraction of privacy-sensitive information from datasets or model parameters. Multiple DP mechanisms have been developed to provide guarantees on a certain level of DP, including the Laplacian mechanism \cite{dwork2006calibrating}, Gaussian mechanism \cite{Dwork2014DP}, functional mechanism \cite{zhang2012functional}, and exponential mechanism\cite{mcsherry2007mechanism}. 
The Laplacian, functional, and exponential mechanisms enforce $\epsilon$-DP, which often results in aggressive noise injection for ML tasks \cite{pan2024differential}. The Gaussian mechanism, by contrast, aims to enforce $(\epsilon,\delta)$-DP (see Sec.~\ref{ssec:DP}), a relaxed version of $\epsilon$-DP that is more suitable for learning tasks.


DP has been introduced into FL as a means of preventing servers and external eavesdroppers from extracting private information.
This has traditionally followed two paradigms: (i) central DP (CDP), involving noise addition at the main server~\cite{kon2017federated,Xiong2022CDP}, and (ii) local DP (LDP), which adds noise at each edge device~\cite{Zhao2021LDP,Shen2022imp,mobi2023Qiao,Liu2023mobi}. 
CDP generally leads to a more accurate final model, but it hinges on the trustworthiness of the main server. Conversely, LDP forgoes this trust requirement but requires a higher level of noise addition at each device to compensate~\cite{naseri2022local}.

There have been a few recent works dedicated to integrating these two paradigms into HFL, i.e., the special case of MFL with two layers. Specifically, \cite{Shi2021HDP,Zhou2023HDL} adapted the LDP strategy to the HFL structure, utilizing moment accounting to obtain strict privacy guarantees across the system. \cite{Wainakh2020HLDP} explored the advantages of flexible decentralized control over the training process in HFL and examined its implications on participant privacy. 
More recently, a third paradigm called hierarchical DP   (HDP) has been introduced~\cite{chandrasekaran2024hierarchical}. 
HDP assumes that intermediate nodes present within the network can be trusted even if the main server cannot. These nodes are assigned the task of adding calibrated DP noise to the aggregated models prior to passing them upstream. The post-aggregation DP addition requirement to meet a given privacy budget becomes smaller as a result. However, there has not yet been a comprehensive analytical study on HDP and the convergence behavior in these systems, without which control algorithm design remains elusive. Moreover, no work has attempted to extend the concept of HDP towards more general multi-tier networks found in MFL settings, where the challenges of noise injection coupling and heterogeneous trust models become further exacerbated.

\subsection{Outline and Summary of Contributions}
In this work, we bridge this gap through the development of \textit{\underline{M}ulti-Tier \underline{F}ederated Learning with \underline{M}ulti-Tier \underline{D}ifferential \underline{P}rivacy} ({\tt M$^2$FDP}), a privacy-preserving model training paradigm which fuses a flexible set of trusted servers in HDP with MFL procedures. Our convergence analysis reveals conditions necessary to secure robust convergence rates in DP-enhanced MFL systems, providing a foundation for designing control algorithms to adapt the tradeoff among energy consumption, training delay, model accuracy, and data privacy. Our main contributions can be summarized as follows:
\begin{itemize}[leftmargin=5mm]
    \item We formalize {\tt M$^2$FDP}, which solves the important challenge of determining how to inject noise into MFL according to the privacy-utility tradeoff. {\tt M$^2$FDP} integrates flexible Multi-Tier DP (MDP) trust models and noise injection with MFL systems comprised of an arbitrary number of network layers (Sec.~\ref{sec:prelim}\&\ref{sec:tthf}). We design this methodology to preserve a target privacy level throughout the entire training process, instead of only at individual aggregations, allowing for a more effective balance between privacy preservation and model performance. {\color{black} Further discussion shows our algorithm is able to recover existing DP-FL and Fed-Avg based algorithms under specific configurations of network layers and trustworthy nodes.}
    
    \item We characterize the convergence behavior of {\tt M$^2$FDP} under non-convex ML loss functions (Sec.~\ref{sec:convAnalysis}). Our analysis (culminating in Theorem~\ref{thm:noncvx}) shows that with an appropriate choice of FL step size,
    the cumulative average global model will converge sublinearly with rate $\mathcal O(1/\sqrt{T})$ to a region around a stationary point. The stationarity gap depends on factors including the trust model, multi-tier network layout, and aggregation intervals. 
    These results also reveal how DP noise injected above certain tiers induces smaller degradation on ML convergence performance compared to injection at lower tiers (Corollary~\ref{cor:noncvx}). {\color{black} We also conduct a comprehensive privacy analysis of {\tt M$^2$FDP}, showing consistent DP protection across all nodes in the multi-tier network.}
    
    \item We demonstrate an important application of our convergence analysis: developing an adaptive control algorithm for {\tt M$^2$FDP} (Sec.~\ref{sec:ctrl_DPFL}). This algorithm simultaneously optimizes energy consumption, latency, and model training performance, while maintaining a sub-linear convergence rate and desired privacy standards. This is achieved through fine-tuning the local training interval length, learning rate, and fraction of edge devices engaged in FL together with the DP injection profile during each training round. 
    
    \item Through numerical evaluations, we demonstrate that {\tt M$^2$FDP} obtains substantial improvements in convergence speed and trained model accuracy relative to existing DP-based FL algorithms (Sec.~\ref{sec:experiments}). Further, we find that the control algorithm reduces energy consumption and delay in multi-tier networks by up to $60\%$. Our results also corroborate the impact of the network configuration and trust model on training performance found in our bounds.
\end{itemize}

An abridged version of this paper published in the 2025 IEEE International Conference on Communications. Compared to the conference version, this paper makes the following additional contributions: (1) development of the adaptive control algorithm based on our analysis; (2) a more comprehensive set of experiments, with results on different datasets, network configurations, and privacy settings; and (3) inclusion of sketch proofs for key results in the main text, with detailed proofs deferred to the supplementary material.

\begin{figure*}[t]
  \centering
    \includegraphics[width=0.95\textwidth]{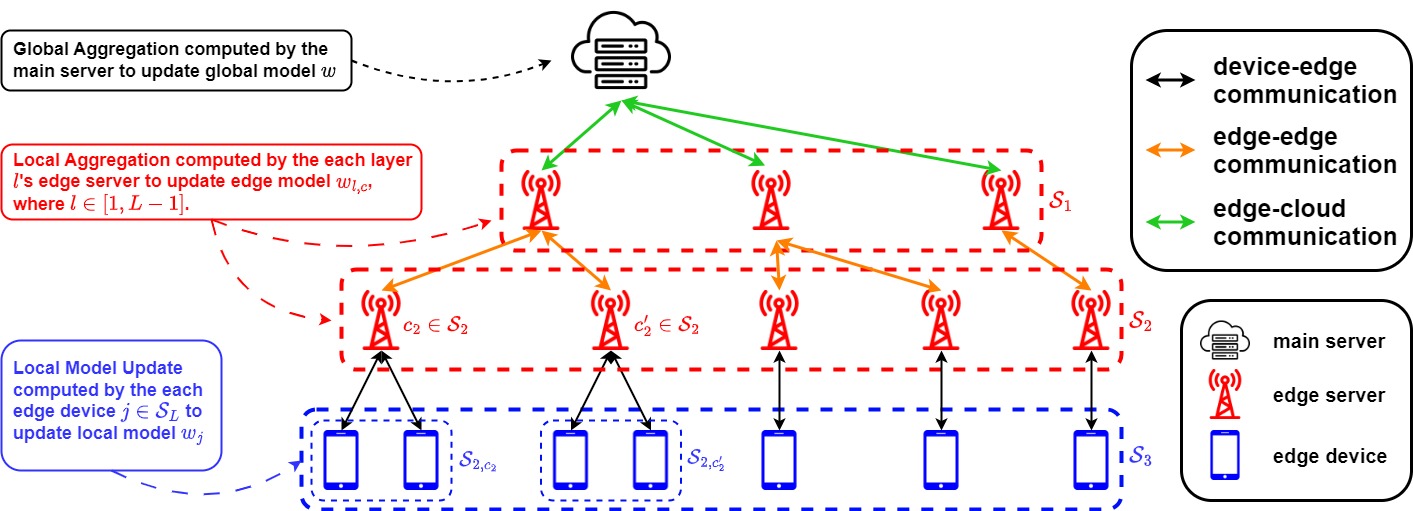}
     \caption{Multi-tier FL network architecture studied in this work with total $L = 3$ layers beneath the main server. $\mathcal{S}_1$ and $\mathcal{S}_2$ are the set of edge servers between the devices and the main server, and $\mathcal{S}_3$ is the set of edge devices computing local gradient updates. The edge devices are further grouped into children nodes of layer $l = 2$, e.g., $\mathcal{S}_{2,c_2}$ and $\mathcal{S}_{2,c'_2}$.}
     \label{fig2}
     \vspace{-1.5em}
\end{figure*}

\vspace{-0.10in}
\section{Preliminaries and System Model}
\label{sec:prelim}
\noindent This section introduces key DP concepts (Sec.~\ref{ssec:DP}), our hierarchical network model (Sec.~\ref{subsec:syst1}), and the target ML task (Sec.~\ref{subsec:syst2}). Key notations we use throughout the paper are summarized in Table~\ref{table:1}.

\vspace{-0.10in}
\subsection{Differential Privacy (DP)}\label{ssec:DP}
Differential privacy (DP) characterizes a randomization technique according to parameters $\epsilon, \delta$. Formally, a randomized mechanism $\mathcal M$ adheres to ($\epsilon$,$\delta$)-DP if it satisfies the following: 
\begin{definition}[($\epsilon$,$\delta$)-DP~\cite{Dwork2014DP}]
For all $\mathcal D$, $\mathcal D'$ that are adjacent datasets, and for all $\mathcal S \subseteq \text{Range}(\mathcal M)$, it holds that:
\begin{align}
    \Pr[\mathcal M(D)\in\mathcal S]\leq e^{\epsilon}\Pr[\mathcal M(D')\in\mathcal S]+\delta,
\end{align}
where $\epsilon>0$ and $\delta\in(0,1)$. 
\end{definition}
\noindent $\epsilon$ represents the privacy budget, quantifying the degree of uncertainty introduced in the privacy mechanism. Smaller $\epsilon$ implies a stronger privacy guarantee. $\delta$ bounds the probability of the privacy mechanism being unable to preserve the $\epsilon$-privacy guarantee, and adjacent dataset means they differ in at most one element.


\textbf{Gaussian mechanism~\cite{Dwork2014DP}}:
The Gaussian mechanism is a commonly employed randomization mechanism compatible with $(\epsilon,\delta)$-DP. With it, noise sampled from a Gaussian distribution is introduced to the output of the function being applied to the dataset. This function, in the case of {\tt M$^2$FDP}, is the computation of gradients.

Formally, to maintain ($\epsilon$,$\delta$)-DP for any query function $f$ processed utilizing the Gaussian mechanism, the standard deviation $\sigma$ of the distribution must satisfy:
\begin{align}
    \sigma > \frac{\sqrt{2log(1.25/\delta)} \Delta f}{\epsilon},
\end{align}
where $\Delta f$ is the $L_2$-sensitivity, and $\epsilon$, $\delta$ are the privacy parameters.

\textbf{$L_2$-Sensitivity~\cite{Dwork2014DP, Sun2022DP}}:
The sensitivity of a function is a measure of how much the output can change between any given pair of adjacent datasets. Formally, the $L_2$-sensitivity for a function $f$ is defined as:
\begin{align}
    \Delta f = \max_{\mathcal D, \mathcal D'} \quad&\Vert f(\mathcal D)-f(\mathcal D')\Vert_2,
\end{align}
where the maximum is taken over pairs of adjacent datasets $\mathcal{D}$ and $\mathcal{D'}$, and $\Vert\cdot\Vert_2$ is the $L_2$ norm.
In our setting, $L_2$ sensitivity allows calibrating the amount of Gaussian noise to be added to ensure the desired ($\epsilon$,$\delta$)-differential privacy in FL model training.


\begin{table}
    \caption{Summary of key notations used throughout the paper.}
    \centering
    \begin{tabular}{m{0.07\textwidth}|m{0.37\textwidth}}
    \hline
    \centering Notation   &  \centering Description \tabularnewline
    \hline
    \centering $\epsilon$ & The privacy budget that quantifies the uncertainty level of privacy protection.\\
    \hline
    \centering $\delta$ & The additional probability bound on privacy mechanism that loosens the $\epsilon$-privacy guarantee restriction.\\
    \hline
    \centering $\Delta f$ & The $L_2$-sensitivity of a function $f$.\\
    \hline
    \centering $L$  & The total number of layers (tiers) in the system below the cloud server. The set of edge servers connected to the global server are labeled as the first layer, and the set of edge devices are labeled as the $L$th layer.\\
    \hline
    \centering $\mathcal{S}_l$ & The set of nodes on the layer $l\in [1, L]$.\\
    \hline
    \centering $\mathcal{S}_{l,c}$ & For a given node at layer $l$, $c \in \mathcal{S}_l$, the set of child nodes in layer $l+1$ connecting to node $c$.\\
    \hline
    \centering {\color{black} $\Delta_{l, c_l}$} & {\color{black} The required $L_2$-sensitivity for an aggregated model at layer $l \in [1, L-1]$ for any given node $c_l \in \mathcal{S}_l$.}\\
    \hline
    \centering $N_l$ & The number of nodes on each layer $l$.\\
    \hline
    \centering $\mathcal{N}_{T,l}$ & The set of secure/trusted edge servers at layer $l$, $l \in [1, L-1]$.\\
    \hline
    \centering $\mathcal{N}_{U,l}$ & The set of insecure/untrusted edge servers at layer $l$, $l \in [1, L-1]$.\\
    \hline
    \centering $T$ & The total number of global aggregations from edge devices toward the main server throughout the whole training process.\\
    \hline
    \centering $K^t$ & The total number of local gradient updates between global iteration $t$ and $t+1$.\\
    \hline
    \centering $K^{\max}$ & The largest $K^t$ throughout the whole training process.\\
    \hline
    \centering $\mathcal{K}_l^t$ & The set of local iterations between global iteration $t$ and $t+1$, where all edge devices perform local aggregation towards layer $l$, $l \in [1, L-1]$.\\
    \hline 
    \centering $\mathcal{D}_i$ & The dataset allocated on the edge device $i \in \mathcal{S}_L$.\\
    \hline 
    \centering $w^{(t)}$ & The global model at global iteration $t$.\\
    \hline 
    \centering $w_{l,c}^{(t,k)}$ & The model aggregated to edge server $c \in \mathcal{S}_l, l\in [1, L-1]$ at global iteration $t$ and local iteration $k$.\\
    \hline
    \centering $w_j^{(t,k)}$ & The model computed at edge device $j \in \mathcal{S}_L$ at global iteration $t$ and local iteration $k$.\\
    \hline
    \centering $g_i^{(t,k)}$ & The unbiased stochastic gradient computed on edge device $i \in \mathcal{S}_L$ during global iteration $t$ and local iteration $k$.\\
    \hline
    \centering $d_{l,j}$ & A edge server in $\mathcal{S}_l$ that is the parent node of the edge device $j \in \mathcal{S}_L$ at layer $l$.\\
    \hline  
    \centering $\widetilde{w}_{l, c}^{(t,k)}$ & An intermediate model at edge server $c \in \mathcal{S}_l$ where additional noise is injected for aggregation towards an edge server in layer $l-1$.\\
    \hline
    \centering $\widetilde{w}_{j}^{(t,k)}$ & An intermediate model at edge devices $j \in \mathcal{S}_L$ where gradient update is being performed but local aggregation and broadcasting is not performed yet.\\
    \hline
    \centering $\rho_{l,i,j}$ & Aggregation weights applied onto local aggregations from $j \in \mathcal{S}_{l,i}$ towards $i \in \mathcal{S}_l$.\\
    \hline
    \centering $\rho_{c}$ & Aggregation weights applied when a edge server $c\in \mathcal{S}_1$ at layer $l=1$ aggregates to the main server.\\
    \hline 
    \centering $n_{l,c}^{(t,k)}$ & Noise attached to model $w_{l,c}^{(t,k)}$ during aggregation to protect differential privacy, $l \in [1, L]$.\\
    \hline
    \centering $\Delta_{l, c}$ & The L2-sensitivity to generate the Gaussian noise $n_{l,c}^{(t,k)}$ during aggregations.\\
    \hline 
    \centering $p_{l,c}$ & The ratio of child nodes for some $c \in \mathcal{N}_{U,l-1}$ that are secure servers.\\
    \hline 
    \centering $p_l^{\max}/p_l^{\min}$ & The maximum/minimum $p_{l,c}$ throughout all $c \in \mathcal{N}_{U,l-1}$.\\
    \hline
    \centering $s_l$ & The smallest set of child nodes out of all edge servers $c \in \mathcal{S}_l$.\\
    \hline
    \centering \textcolor{black}{$\mathcal{A}_l$} & \textcolor{black}{DP effects at layer $l$ from DP noise injected at layer $l$, without injection at layers $l' > l$.}\\
    \hline
    \centering \textcolor{black}{$\mathcal{B}_l$} & \textcolor{black}{DP effects at layer $l$ from DP noise injection at layers $l' \in [l+1, L-1]$.}\\
    \hline \centering \textcolor{black}{$\mathcal{C}_l$} & \textcolor{black}{DP effects at layer $l$ from DP noise injection at edge devices (layer $L$).}\\
    \hline
    \end{tabular}
    \vspace{-0.10in}
    \label{table:1}
\end{table}

\subsection{Multi-Tier Network System Model}
\label{subsec:syst1} 

\textbf{System architecture}: We consider the multi-tier edge/fog network architecture depicted in Fig.~\ref{fig2}. We define $L$ as the number of layers (i.e., tiers) below the cloud server, i.e., the total number of layers\textcolor{black}{, including the cloud server,} is $L+1$. The hierarchy, from bottom to top, consists of edge devices at the $L$th layer, $L-1$ layers of intermediate nodes/servers, and the cloud server. The primary responsibilities of these layers include local model training (at edge devices), local/intermediate model aggregations (at intermediate fog nodes), and global model aggregations (at the cloud server).

The set of nodes in each layer $l$ is denoted by $\mathcal{S}_l$, where $\mathcal{S}_{l,c}$ is the set of child nodes in layer $l+1$ connecting to node $c $ in $\mathcal{S}_l$. The sets $\mathcal{S}_{l,c}$ are all pairwise disjoint, meaning every node in layer $l+1$ is connected to a single upper layer node in layer $l$. We define $N_l = |\mathcal{S}_l| $ as the total number of nodes in layer $l$. The total number of edge devices is $N_L$.


\begin{figure*}[t]
\includegraphics[width=1.0\textwidth]{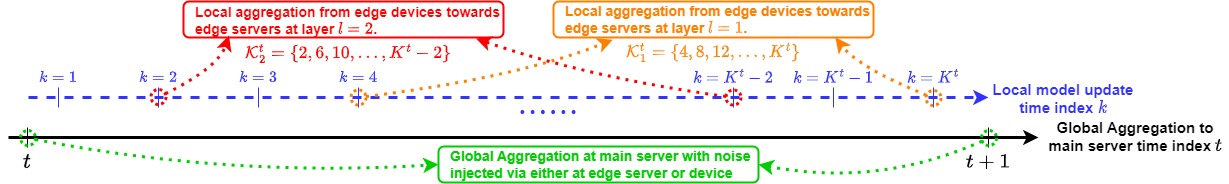}
\centering
\caption{Illustration of training timescales in {\tt M$^2$FDP} for the case of $L=3$ in Fig.~\ref{fig2}. $K^t$ total local gradient updates are performed between global iterations $t$ and $t+1$. Within the time interval of local iterations $k \in [1, K^t]$, multiple local aggregations from edge devices towards edge servers at layers $l=1$ (i.e., $\mathcal{K}_1^t$) and $l=2$ (i.e., $\mathcal{K}_2^t$) are performed.}
\label{fig:twoTimeScale}
\vspace{-0.15in}
\end{figure*}

\textbf{Threat model}:
We consider the setting where any information that is being passed through a communication link runs the risk of being captured, but there is no risk of any information being poisoned or replaced. This includes scenarios such as external eavesdroppers, or when certain intermediate nodes exhibit \textit{semi-honest behavior}~\cite{Melis2019TM,Zhu2019,Wang2019TM}, where the intermediate nodes still perform operations as the FL protocol intended, but they may seek to extract sensitive information from shared FL models. For instance, a network operator may be providing service through a combination of its own infrastructure (trusted/secure) as well as borrowed infrastructure (untrusted/insecure). 

Each node $c\in \mathcal{S}_l$, for $l \in [1,L-1]$, will periodically conduct model aggregations among its child nodes $\mathcal{S}_{l,c}$. The hierarchical nature of this operation enforces the following trustworthiness designations among nodes:
\begin{itemize}
\item Node $c$ can be categorized as a \textit{secure/trusted} intermediate node only if it is deemed trustworthy by all its children, i.e., $c \in \mathcal{N}_{T,l} \subseteq \mathcal{S}_l$. 

\item On the other hand, if not all children trust node $c$, then node $c$ is automatically categorized as an \textit{insecure/untrusted} node, i.e., $c \in \mathcal{N}_{U,l} \subseteq \mathcal{S}_l$.

\item If any of node $c$'s children have been labeled insecure by these mechanisms, i.e., $\mathcal{S}_{l,c} \cap \mathcal{N}_{U,l+1} \neq \emptyset$, then node $c$ is automatically labeled as untrusted (as it will receive models from such nodes). 
\end{itemize}
As a consequence, all insecure nodes in layer $l$ are children of insecure nodes in layer $l-1$:
\begin{align}
    \mathcal{N}_{U,l} \subseteq \bigcup_{c \in \mathcal{N}_{U, l-1}} \mathcal{S}_{l-1,c}.
\end{align}
%
\textcolor{black}{These trustworthiness labels of nodes -- i.e., whether a node $c$ in layer $l$ is in $\mathcal{N}_{T,l}$ or $\mathcal{N}_{U,l}$ -- are assumed to be set prior to the operation of our {\tt M$^2$FDP} method. These labels may, for example, reflect relationships between infrastructure owners or service providers (e.g., if there are intermediate nodes in the hierarchy not owned/managed by the entity responsible for the learning task). We make no other assumptions on the model employed for determining trustworthiness of nodes.}

\vspace{-0.10in}
\subsection{Machine Learning Model} \label{subsec:syst2}
Each edge device $j\in \mathcal{S}_L$ has a dataset $\mathcal{D}_j$ comprised of $D_j=|\mathcal{D}_j|$ data points. We consider the datasets $\mathcal{D}_j$ to be non-i.i.d. across devices, as is standard in FL research.
The \textit{loss} $\ell(d; w)$ quantifies the fit of the ML model to the learning task. It is linked to a data point $d \in \mathcal{D}_j$ and depends on the 
 model $w \in \mathbb{R}^M$, with $M$ representing the model's dimension.
Consequently, the \textit{local loss function} for device $i$ is:
\begin{align}\label{eq:1}
\textstyle F_j(w)=\frac{1}{D_j}\sum_{d\in\mathcal D_j}
\ell(d; w).
\end{align}
We define the \textit{intermediate loss function} for each $ c\in \mathcal{S}_{l}$ as
\begin{align*}
\textstyle F_{l,c}(w) = \begin{cases}
    \sum_{j\in \mathcal{S}_{l,c}} \rho_{l,c,j} F_j(w), &\quad l = L-1,\\
    \sum_{i\in \mathcal{S}_{l,c}} \rho_{l,c,i} F_{l+1,i}(w), &\quad l \in [1, L-2],
\end{cases}
\end{align*}
where $\rho_{l,c, i} = 1/|\mathcal{S}_{l,c}|$ symbolizes the relative weight when aggregating from layer $l+1$ to $l$.
Finally, the \textit{global loss function} is defined as the average loss across all subnets:
\begin{align*}
\textstyle F(w)=\sum_{c=1}^{N_1} \rho_c F_{1,c}(w),
\end{align*}
where $\rho_c = 1/N_1$ is each subnet's weight in the global loss.

The learning objective of the system is to determine the optimal global model parameter vector $w^* \in \mathbb{R}^M$ such that
$
w^* = \mathop{\argmin_{w \in \mathbb{R}^M} }F(w).
$

\vspace{-0.10in}
\section{Proposed Methodology}
\label{sec:tthf}
\noindent In this section, we formalize {\tt M$^2$FDP}, including its timescales of operation (Sec. \ref{subsec:syst3}), training process (Sec. \ref{subsec:form}), and DP mechanism (Sec. \ref{ssec:DP_main}).

\vspace{-0.10in}
\subsection{Model Training Timescales} \label{subsec:syst3}
Training in {\tt M$^2$FDP} follows a slotted-time representation, depicted in Fig.~\ref{fig:twoTimeScale}. \textit{Global aggregations} are carried out by all layers of the hierarchy to collect model parameters for the cloud server at time indices $t = 1, ..., T$. In between two global iterations $t$ and $t+1$, $K^t$ \textit{local model training iterations} are carried out by the edge devices. We assume these local update iterations occur via stochastic gradient descent (SGD).



In {\tt M$^2$FDP}, the main server first broadcasts the initial global model ${w}^{(1)}$ to all devices at $t = 1$. For each global iteration $t$, we define $\mathcal{K}_l^t \subseteq \{1, 2, \ldots, K^t\}$ as the set of local iterations for which a \textit{local aggregation} is conducted by parent nodes in layer $l$. We assume the sets of local aggregation iterations, $\mathcal{K}_1^t, \mathcal{K}_2^t, \ldots, \mathcal{K}_{L-1}^t$ are pairwise disjoint, i.e., if local aggregations occur at iteration $k$, they are dictated by one layer.

\begin{algorithm}[t]
{ \small
\caption{\strut\small {\tt M$^2$FDP}: Multi-Tier Federated Learning with Multi-Tier Differential Privacy}
\label{alg:1}
\KwIn{Number of global aggregations $T$, minibatch sizes $\vert\xi_i^{(t,k)}\vert$, target DP level $(\epsilon, \delta)$} 
\KwOut{Global model ${ w}^{(T+1)}$}
Initialize $w^{(1)}$ and broadcast it among the edge devices through the intermediate nodes, set $w_j^{(1,1)} = w^{(1)}, \forall j \in \mathcal{S}_L$.\\
Initialize $K^1$ to decide the number of local SGD steps to be performed before the first global aggregation.\\
Initialize $\eta^1$ according to the conditions on step size in Theorem~\ref{thm:noncvx}.

\For{$t = 1, \ldots, T$}{
\For{$k = 1, \ldots, K^t$}{
Each edge device $i \in \mathcal{S}_L$ performs a local SGD update based on \eqref{eq:SGD} using $w_i^{(t,k)}$ to obtain $\widetilde{w}_j^{(t,k)}$.\\
\uIf{$k \in \mathcal{K}^t_l$ for some $l\in [1, L-1]$}{
\For{$l' \in [L-1, l]$}{
\NodeFor{$c \in \mathcal{S}_{l'}$}{Children $\mathcal{S}_{l',c}$ upload their models to parent $c$, adding noise based on~\eqref{eq:noiseEdge}\&\eqref{eq:noise} if $c \in \mathcal{N}_{U,l'}$. \\
\uIf{$c$ is secure, i.e., $c\in \mathcal{N}_{T,l'}$}{
Local aggregation at $c$ follows \eqref{eq:secure_local_aggr}. }
\Else{Local aggregation at $c$ follows \eqref{eq:insecure_local_aggr}.}}
}
Nodes $c \in \mathcal{S}_l$ broadcast their aggregated models $w_{l,c}^{(t,k)}$ to the edge, resulting in $w_j^{(t,k)} = {w}_{l, d_{l,j}}^{(t,k)}, \forall j \in \mathcal{S}_L$, where $d_{l,j}$ denotes $j$'s ancestor in layer $l$.

}
\Else{
${{{w}}}_j^{(t,k)} = \widetilde{w}_j^{(t,k)}, \quad j\in\mathcal{S}_L.$
}

}
Main server conducts global aggregation to compute the global model ${ w}^{(t+1)}$ based on \eqref{eq:secure_local_aggr}, \eqref{eq:insecure_local_aggr}, and \eqref{eq:glob_aggr_3}.\\
Broadcast the aggregated model to edge devices for the next round of local iterations $w_i^{(t+1,1)} = { w}^{(t+1)} \; \forall i \in \mathcal{S}_L.$

}

}

\end{algorithm}
\setlength{\textfloatsep}{0pt}
\vspace{-0.10in}
\subsection{{\tt M$^2$FDP} Training and Aggregations}\label{subsec:form}
Algorithm~\ref{alg:1} summarizes the full {\tt M$^2$FDP} procedure. It can be divided into three major stages: (1) local model updates using on-device gradients, (2) DP noise injection and local model aggregations at intermediate nodes, and (3) global model aggregations in the cloud, as follows.

\textbf{Local model update}: At local iterations $k \in [1, K^t]$, device $j$ randomly selects a mini-batch $\xi_j^{(t,k)}$ from its dataset $\mathcal D_j$. Using this mini-batch, it calculates the stochastic gradient
\begin{equation}
    {g}_{j}^{(t,k)}=\frac{1}{\vert\xi_j^{(t,k)}\vert}\sum_{d\in\xi_j^{(t,k)}}
    \nabla\ell(d; w_j^{(t,k)}).
\end{equation}
We assume a uniform selection probability $q$ of each data point, i.e., $q=\vert\xi_j^{(t)}\vert/D_j,~\forall j$. Device $j$ employs ${ g}_{j}^{(t,k)}$ to determine $ {\widetilde{{w}}}_j^{(t)}$:
\begin{align} \label{eq:SGD}
    \textstyle\widetilde{w}_j^{(t,k+1)} = 
            w_j^{(t,k)}-\eta^t { g}_{j}^{(t,k)},~k \in [1,K^t],
\end{align}
where $\eta^{t} > 0$ is the step size. Using ${\widetilde{w}}_j^{(t,k+1)}$ as the \textit{base} local model, the \textit{updated} ${w}_j^{(t,k+1)}$ is determined in one of several ways depending on the trust model.
Specifically, if no local aggregation is performed at time $k$, i.e., $k \notin \cup_{l=1}^{L-1}\mathcal{K}_l^t$, the updated model follows ${{{w}}}_j^{(t,k+1)} = {\widetilde{{w}}}_j^{(t,k+1)}$ in \eqref{eq:SGD}. On the other hand, if $k \in \cup_{l=1}^{L-1}\mathcal{K}_l^t$, then the updated local model inherits the local model aggregation described next.

\textbf{Local model aggregations with DP noise injection}: For an iteration $k \in \mathcal{K}_l^t$, the network performs aggregations for nodes up through layer $l$. First, the \textit{local aggregated model} $w_{l',c}^{(t,k)}$ will be computed at the intermediate nodes $c \in \mathcal{S}_{l'}$ in level $l' = L-1$. These will be passed up to to layer $l' = L-2$, with the aggregations continuing sequentially until arriving at layer $l' = l$. This involves one of two operations, depending on the trustworthiness labels of the intermediate nodes:

\textit{(i) Aggregation at secure nodes:} For aggregations performed at layer $l'$, if aggregator $c\in \mathcal{S}_{l'}$ is labeled secure, i.e., $c \in \mathcal{N}_{T,l'}$, each child node $i \in \mathcal{S}_{l',c}$ uploads its model parameters without noise added. Hence, ${w}_{l',c}^{(t,k)}$ is computed as:
\vspace{-0.05in}
\begin{align}\label{eq:secure_local_aggr}
 w_{l',c}^{(t,k)} &= \begin{cases}
    \sum_{i\in \mathcal{S}_{l',c}}\rho_{l',c,i}\widetilde{w}_i^{(t,k)},\quad& l' = L-1,\\
    \sum_{i\in \mathcal{S}_{l',c}}\rho_{l',c,i}{w}_{l'+1,i}^{(t,k)} \quad& l' \in [l, L-2].\\
\end{cases}
\end{align} 

\textit{(ii) Aggregation at insecure nodes:} Conversely, if $c \in \mathcal{S}_{l'}$ is labeled untrustworthy, i.e., $c\in\mathcal{N}_{U,l'}$, each child node $i \in\mathcal{S}_{l',c}$ 
injects DP noise $n_{l'+1,i}^{(t,k)}$ within its transmission, i.e.,:

\begin{figure*}[t]
  \centering
    \includegraphics[width=0.95\textwidth]{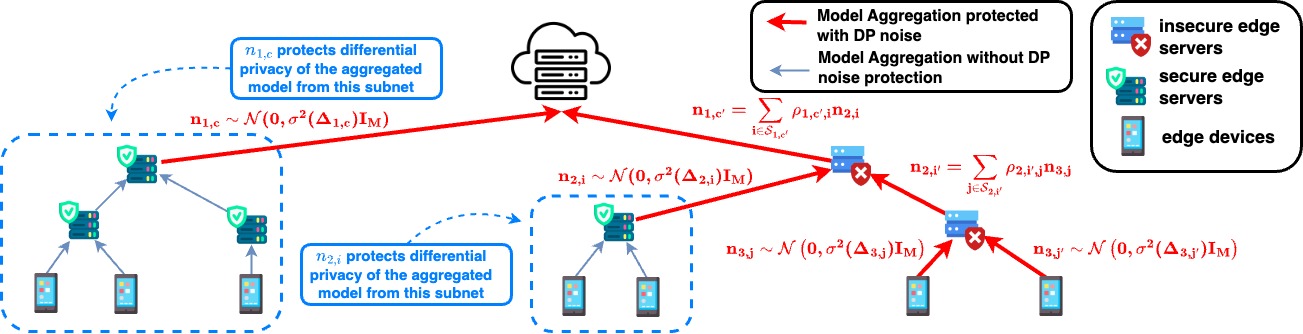}
     \caption{Illustration of how noise is injected through aggregation to ensure protection on insecure servers. During local aggregation at each layer $l$, the noise added towards the child node's model will be a linear combination of existing noises if the child node is in the set of insecure edge servers $\mathcal{N}_{U,l+1}$. Otherwise, if the child node is in the set of secure edge servers $\mathcal{N}_{T,l+1}$, a new noise that provides the target DP criterion will be generated and injected to the child node's model.}
     \label{fig:DP}
   \vspace{-0.2in}
\end{figure*}

{\small
\begin{align}\label{eq:insecure_local_aggr}
 \widetilde{w}_{l,i}^{(t,k)} &= \begin{cases}
    {w}_{l',i}^{(t,k)}, &\quad i \in \mathcal{N}_{U,l},\\
    {w}_{l',i}^{(t,k)} + n_{l',i}^{(t,k)}, &\quad i \in \mathcal{N}_{T,l},
\end{cases}\\
 w_{l',c}^{(t,k)} &= \begin{cases}
    \sum_{j\in \mathcal{S}_{l',c}}\rho_{l',c,j}\left(\widetilde{w}_j^{(t,k)} + n_{L,j}^{(t,k)}\right), \quad& l' = L-1,\\
    \sum_{i\in \mathcal{S}_{l',c}}\rho_{l',c,i}\widetilde{w}_{l'+1,i}^{(t,k)}, \quad& l' \in [l, L-2].\notag
\end{cases}
\end{align} 
}
For all edge devices $j \in \mathcal{S}_L$, the noise is sampled from:
\vspace{-1mm}
{\small
\begin{align}
    \textstyle n_{L,j}^{(t,k)} &\sim \mathcal{N}(0, \sigma^2(\Delta_{L,j})I_M),\notag\\
    \Delta_{L,j} &\overset{\Delta}{=} \max_{\mathcal{D}, \mathcal{D}'}\left\|\eta^t\sum_{k=1}^{K^t}\left(g_j^{(t,k)}(\mathcal{D}) - g_j^{(t,k)}(\mathcal{D}')\right)\right\|, \label{eq:noiseEdge}
\end{align}
}

\noindent where $\sigma(\cdot)$ is a function of $\Delta$ defined later in Proposition~\ref{prop:GM}.
On the other hand, for all intermediate nodes $i \in \mathcal{N}_{T,l'}, l' \in [l+1, L-1]$, the noise is sampled from:

\vspace{-3mm}

{\small
\begin{align}
    \textstyle n_{l',i}^{(t,k)} &\textstyle \sim \mathcal{N}(0, \sigma^2(\Delta_{l',j})I_M)\notag\\
    \textstyle \Delta_{l',i} &\textstyle \overset{\Delta}{=} \max_{\mathcal{D}, \mathcal{D}'}\Big\|\eta^t\sum_{k=1}^{K^t}\sum_{c_{l'+1}\in \mathcal{S}_{l',i}}\rho_{l, i, c_{l'+1}} \ldots \label{eq:noise} \\
    &\textstyle \sum_{j\in \mathcal{S}_{L-1, c_{L-1}}} \rho_{L-1, c_{L-1}, j} \left(g_j^{(t,k)}(\mathcal{D}) - g_j^{(t,k)}(\mathcal{D}')\right)\Big\|.\notag
    \end{align}
}

\noindent Examples of such noise additions are illustrated in Fig.~\ref{fig:DP}.

Finally, after computing the local aggregated model, node $c$ at layer $l$ broadcasts ${{{w}}}_{l,c}^{(t,k)}$ across its children, continuing sequentially to the edge. Formally, letting $d_{l,j} \in \mathcal{S}_l$ be the ancestor node of edge device $j \in \mathcal{S}_L$ at layer $l$, the devices will synchronize their local models as \begin{equation}
    {{{w}}}_j^{(t,k)} = {w}_{l, d_{l,j}}^{(t,k)}, \quad\forall j\in\mathcal{S}_L.
\end{equation} 

\textbf{Global model aggregation with DP noise injection}: 
Global aggregation occurs at the end of each training interval $K^t$. First, an intermediate aggregation at layer $l = 1$ is conducted through the procedure in \eqref{eq:secure_local_aggr}, \eqref{eq:insecure_local_aggr}. Then, for the aggregation at the main server, all models from intermediate nodes of layer $l = 1$ will be protected due to the assumption that the cloud is insecure.

For secure nodes $c \in \mathcal{N}_{T,1}$, noise $n_{1, c}^{(t,K^t+1)} \sim \mathcal{N}(0, \sigma^2(\Delta_{1,c})I_M)$ will be added as in~\eqref{eq:noise}. For insecure nodes $c \in \mathcal{N}_{U,1}$, no noise will be added since noises have already been attached during intermediate aggregations. Thus,
\begin{align}
\label{eq:glob_aggr_3}
\widetilde{w}_{1,c}^{(t, K^t+1)} &= \begin{cases}
        {w}_{1,c}^{(t, K^t+1)}, &\quad c \in \mathcal{N}_{U,1},\\
        {w}_{1,c}^{(t, K^t+1)} + n_{1,c}^{(t,K^t+1)}, &\quad c \in \mathcal{N}_{T,1},
    \end{cases}\notag\\
    w^{(t+1)} &= \sum_{c = 1}^{N_1} \rho_{c} \widetilde{w}_{1,c}^{(t, K^t+1)}.
\end{align}
Upon completion of the calculations at the main server, the resulting global model ${ w}^{(t+1)}$ is employed to synchronize the local models maintained by the edge devices, i.e., ${{{w}}}_i^{(t+1,1)} = { w}^{(t+1)} \; \forall i \in \mathcal{S}_L$.

Finally, the resulting global model ${w}^{(t+1)}$ is employed to synchronize the local models across all edge devices, i.e., ${{{w}}}_j^{(t+1,1)} = { w}^{(t+1)}, \forall j \in \mathcal{S}_L$.

\vspace{-0.10in}
\subsection{DP Mechanisms}\label{ssec:DP_main}
We now formalize the procedure for configuring the DP noise variables. In this study, we focus on the Gaussian mechanisms from Sec.~\ref{ssec:DP}, though {\tt M$^2$FDP} can be adjusted to accommodate other DP mechanisms too. 
 
Following the composition rule of DP~\cite{Dwork2014DP}, we aim to take into account the privacy budget \textit{across all aggregations throughout the training}. This will ensure cumulative privacy for the complete model training process, rather than considering each individual aggregation in isolation~\cite{Poor2020LDP,Shen2022imp}. Below, we define the Gaussian mechanisms, incorporating the moment accountant technique \cite{Shi2021HDP,Zhou2023HDL}.
These mechanisms utilize Assumption~\ref{assump:genLoss} which is stated in Sec.~\ref{sec:convAnalysis}.
\begin{proposition}[Gaussian Mechanism~\cite{Abadi2019MA}]\label{prop:GM}
    Under Assumption~\ref{assump:genLoss}, there exists constants $c_1$ and $\alpha_l$ such that given the data sampling probability $q$ at each device, and \textcolor{black}{the total number of aggregations $T$ conducted during the model training process, for any $\epsilon<c_1qT$},  {\tt M$^2$FDP} exhibits $(\epsilon,\delta)$-differential privacy for any $\delta>0$, so long as the DP noise follows $\mathbf n_{DP}^{(t,k)}\sim \mathcal{N}(0, \sigma^2(\Delta_{l,c})\mathbf{I}_M)$, where
    \begin{align}
        \sigma(\Delta_{l,c}) = \alpha_l\frac{q\Delta_{l,c}\sqrt{\textcolor{black}{T}\log(1/\delta)}}{\epsilon}.
    \end{align}
    Here, $\Delta_{l,c}$ represent the $L_2$-norm sensitivity of the gradients exchanged during the aggregations towards a node $c \in \mathcal{S}_l$.
\end{proposition}
The characteristics of the DP noises introduced during local and global aggregations can be established using Proposition~\ref{prop:GM}. The relevant $L_2$-norm sensitivities can be established in the following lemma.
\begin{lemma}\label{lem:DeltaM}
    Under Assumption~\ref{assump:genLoss}, the $L_2$-norm sensitivity of the exchanged gradients during local and global aggregations can be obtained as:

\begin{itemize}
    \item For any $l \in [1, L-1]$, and any given node $c_l \in \mathcal{S}_l$, the sensitivity is bounded by
    \begin{align}\label{eq:locAgg_egd_sens1}
        &\Delta_{l, c_l} = \max_{\mathcal{D}, \mathcal{D}'}\Bigg\|\eta^t\sum_{k=1}^{K^t}\sum_{c_{l+1}\in \mathcal{S}_{l,c_l}}\rho_{l, c_l, c_{l+1}} \ldots \notag\\
        &\sum_{j\in \mathcal{S}_{L-1, c_{L-1}}} \rho_{L-1, c_{L-1}, j} \left(g_j^{(t,k)}(\mathcal{D}) - g_j^{(t,k)}(\mathcal{D}')\right)\Bigg\|\notag\\
        &\leq \frac{2\eta^t K^t G}{\prod_{l'=l}^{L-1} s_{l'}}, \quad \forall l \in [1, L-1]
    \end{align}
    \item For any given node $j \in \mathcal{S}_L$, the sensitivity is bounded by
    \begin{align}\label{eq:locAgg_egd_sens2}
        &\Delta_{L,j} = \max_{\mathcal{D}, \mathcal{D}'}\left\|\eta^t\sum_{k=1}^{K^t}\left(g_j^{(t,k)}(\mathcal{D}) - g_j^{(t,k)}(\mathcal{D}')\right)\right\| \notag\\&\leq 2\eta^t K^t G
    \end{align}
\end{itemize}
    \vspace{-0.1in}
\end{lemma}
Then, by defining $K^{\max} = \max_{t\in [1,T]}K^t$ to be the largest local training interval throughout the whole training process, the variance of the generated Gaussian noises
${\sigma^2}(\Delta_{l,c})$ can be determined based on Proposition~\ref{prop:GM}.

\vspace{-0.10in}
\section{Convergence Analysis} \label{sec:convAnalysis}
\noindent In this section, we present our analysis of the convergence behavior of {\color{black}{\tt M$^2$FDP}}. The main results are given in Sec.~\ref{ssec:convAvg}, and we provide a sketch proof of these results in Sec.~\ref{ssec:sketch}. \textbf{The full proofs are available in the supplementary material.}

\vspace{-0.10in}
\subsection{Analysis Assumptions and Quantities}
We first establish a few assumptions commonly employed in the literature that we will consider throughout our analysis.

\begin{assumption}[Characteristics of Noise in SGD~\cite{lin2021timescale,Shen2022imp,Zhang2022AS,reddi2021adaptive}] \label{assump:SGD_noise}
    Consider ${\mathbf n}_{i}^{(t)}=\widehat{\mathbf g}_{i}^{(t)}-\nabla F_i(\mathbf w_{i}^{(t)})$ as the noise of the gradient estimate through the SGD process for device $i$ at time $t$. The noise variance is upper bounded by $\sigma^2 > 0$, i.e., $\mathbb{E}_t[\Vert{\mathbf n}_{i}^{(t)}\Vert^2]\leq \sigma^2~\forall i,t$.
\end{assumption}

\begin{assumption}[General Characteristics of Loss Functions \cite{Shen2022imp}, \cite{Zhang2022AS}, \cite{reddi2021adaptive}]\label{assump:genLoss} 
Assumptions applied to loss functions include:
1) The stochastic gradient norm of the loss function $\ell(\cdot)$ is bounded by a constant $G$, i.e.,  
    $\Vert{ g}_{i}^{(t)}\Vert \leq G, ~\forall i, t.$
      2) Each local loss $F_i$ is $\beta$-smooth $\forall i\in\mathcal{I}$, i.e., 
    $\resizebox{1.\linewidth}{!}{$
        \Vert \nabla F_i(\mathbf w_1)-\nabla F_i(\mathbf w_2)\Vert \leq \beta\Vert \mathbf w_1-\mathbf w_2 \Vert, ~\forall \mathbf w_1, \mathbf w_2 \in \mathbb{R}^M.
        $}\hspace{-1.mm}$
This implies $\beta$-smoothness of $\bar{F}_c$ and $F$ as well. 
\end{assumption}
{\color{black}
Boundedness of the stochastic gradient norm is a commonly adopted assumption in the literature on DP for FL~\cite{Shen2022imp,geyer2017differentially} and DP for machine learning more generally~\cite{kang2014adaptive}. It is necessary for computation of the Gaussian noise variance that ensures $(\epsilon,\delta)$-DP addition based on the $L_2$-sensitivity of the gradient in Proposition~\ref{prop:GM}. In practice, when faced with non-smooth loss functions (stemming from e.g., choice of activation in neural networks), this boundedness can be enforced via gradient clipping techniques~\cite{mcmahan2017learning}. Gradient clipping is widely adopted in private learning algorithms to ensure that gradients remain within a fixed range before noise injection.}

Additionally, we establish a few quantities to facilitate our analysis. First, note that during aggregation at an insecure intermediate node, the level of incoming noise depends on whether the models come from secure or insecure children nodes. To capture this, we define $p_{l,c}$ to be the ratio of children for parent node $c \in \mathcal{N}_{U,l-1}$ that are secure servers, i.e.,
$\textstyle p_{l,c} \overset{\Delta}{=} \frac{|\mathcal{N}_{T, l} \cap \mathcal{S}_{l-1, c}|}{|\mathcal{S}_{l-1, c}|},  l \in [1, L-1]$. We then define the maximum and minimum ratio on each layer $l \in [1, L-1]$ as 
\begin{align}
    &p_l^{\max} = \max_{c \in \mathcal{N}_{U,l-1}} p_{l,c}\\
    &p_l^{\min} = \min_{c \in \mathcal{N}_{U,l-1}} p_{l,c}.
\end{align}
\textcolor{black}{For the main server, we define $p_0^{\max} = p_0^{\min} = 0$ if the server is insecure and requires DP protection, and let $p_0^{\max} = p_0^{\min} = 1$ if the server is secure. For edge devices, since no parameters will be aggregated towards those devices, we set $p_L^{\max} = p_L^{\min} = 1$.}

Finally, we define $s_l = \min_{c\in \mathcal{S}_{l}} |S_{l,c}|$ to be the size of the smallest set of child nodes out of all intermediate nodes $c \in \mathcal{S}_l$.




\vspace{-0.10in}
\subsection{General Convergence Behavior of {\tt M$^2$FDP}}
\label{ssec:convAvg}
We now present our main theoretical result, that the cumulative average of the global loss gradient can attain sublinear convergence to a controllable region around a stationary point under non-convex problems.

\begin{theorem}(\textbf{Non-Convex}) \label{thm:noncvx}
        Under Assumptions~\ref{assump:SGD_noise} and \ref{assump:genLoss}, if $\eta^t=\frac{\gamma}{\sqrt{t+1}}$ with $\gamma\leq\min\{\frac{1}{K^{\max}},\frac{1}{T}\}/\beta$, the cumulative average of global gradient satisfies
\begin{align} \label{eq:noncvx_rate}
\small
        &\frac{1}{T}\sum_{t=1}^T \mathbb{E}\left\|\nabla F(w^{(t)})\right\|^2 \leq \underbrace{\frac{2\beta}{\sqrt{T+1}} \mathbb{E}\left[ F(w^{(1)}) -  F(w^{(T+1)})\right]}_\textrm{(a$_1$)} \notag\\
        &+ \underbrace{\frac{K^{\max}\left(G^2\left(1 + \frac{1}{\beta}\right) + \sigma^2\right)}{T}}_\textrm{(a$_2$)}\notag\\
        &+ \underbrace{\frac{8 L M (K^{\max})^4 q^2\log(1/\delta)}{\epsilon^2}\sum_{l=1}^{L}(1 - p_{l-1}^{\min})^2 \bigg(\mathcal{A}_l + \mathcal{B}_l + \mathcal{C}_l\bigg)}_\textrm{(b)},\notag\\
\end{align}
where
\begin{align}
 \label{eq:ABC_terms}
        \mathcal{A}_l &=  p_{l}^{\max}\frac{\alpha_l^2}{\prod_{l' = l}^{L-1} s_{l'}^2}, \notag\\
        \mathcal{B}_l & = \sum_{m = l}^{L-1} \left(\prod_{l' = l}^{m-1}(1 - p_{l'}^{\min})p_{m}^{\max}\right)\frac{\alpha_m^2}{\prod_{l' = l}^{ m-1} s_{l'}\prod_{l'' = m}^{ L-1} s_{l''}^2}, \notag\\
        \mathcal{C}_l &= \prod_{l' = l}^{ L-1}(1 - p_{l'}^{\min})\frac{\alpha_L^2}{\prod_{l' = l}^{ L-1} s_{l'}}.
\end{align}
\end{theorem}

\begin{figure*}[t]
{\small
    \begin{align} \label{eq:ld_ncx}
    &\mathbb{E}\left[ F(w^{(t+1)}) -  F(w^{(t)})\right] \leq \frac{(\eta^tK^t)^2\beta \sigma^2}{2} \underbrace{- \eta^t \sum_{k=1}^{K^t} \mathbb{E}\left\langle \nabla F(w^{(t)}),  \sum_{c_1 = 1}^{N_1} \rho_{c_1} \ldots\sum_{j\in\mathcal{S}_{L-1,c_{L-1}}} \rho_{L-1, c_{L-1},j} \nabla F_j(w_j^{(t,k)})\right\rangle}_\text{(a)} \notag\\
    &+ \underbrace{\frac{(\eta^t)^2 \beta}{2} \mathbb{E}\left[\left\|\sum_{k=1}^{K^t} \sum_{c_1 = 1}^{N_1} \rho_{c_1} \ldots\sum_{j\in\mathcal{S}_{L-1,c_{L-1}}} \rho_{L-1, c_{L-1},j} \nabla F_j(w_j^{(t,k)})\right\|^2\right]}_\text{(b)} + \underbrace{\frac{\beta}{2}\mathbb{E}\left[\left\|n_{DP}^t\right\|^2\right]}_\text{(c)}. 
\end{align}}
\hrule
\vspace{-0.2in}
\end{figure*}


\textcolor{black}{As suggested by Proposition~\ref{prop:GM} and Lemma~\ref{lem:DeltaM}, the variance of the DP noise inserted should scale with the total number of global aggregations $T$ and the number of layers $L$. Additionally, it is also related to the parameter dimension $M$, maximum local update iterations $K^{\max}$, data sampling probability $q$, and privacy guarantee $(\epsilon, \delta)$.} To counterbalance the influence of DP noise accumulation over successive aggregations, we enforce the step-size condition $\eta^t\leq 1/T$ in Theorem~\ref{thm:noncvx} to scale down the DP noise addition by a factor of $T$. \textcolor{black}{As a result, when the number of global aggregations ($T$) increases, terms ($a_1$) and ($a_2$) in~\eqref{eq:noncvx_rate} decrease, while the impact of DP noise in term (b) remains constant.} This strategy steers the bound in~\eqref{eq:noncvx_rate} towards the region denoted by (b). At the same time, it highlights a delicate balance between privacy preservation and training performance: although the condition $\eta^t\leq1/T$ reduces DP noise, it results in a smaller learning rate, which slows {\tt M$^2$FDP} training.

\textcolor{black}{Term (b) in Theorem~\ref{thm:noncvx} emphasizes a key difference from convergence results in conventional FL and HFL: the DP mechanism in in {\tt M$^2$FDP} induces a constant error on the cumulative gradient average. Unlike the gradient noise $\sigma^2$, whose effect is controllable through proper learning-rate selection, the DP noise is added after updates are computed, and thus has a lasting effect that does not diminish with $T$. Similar error terms are seen in prior works on DP-infused FL~\cite{Shen2022imp,Poor2020LDP}}.

In addition, term (b) conveys the beneficial influence of secure intermediate nodes. The groupings $\mathcal{A}_l$, $\mathcal{B}_l$ (for networks with total layers $L > 1$), and $\mathcal{C}_l$ (for network layers $L \geq 1$) break down three differing effects of noise injection on local aggregations at layer $l$. {\color{black}$\mathcal{A}_l$ represents the part of the hierarchy where no noises have been injected up to the aggregation at layer $l$, i.e., the secure branches of the tree. The noise level introduced during aggregations here is reduced by a factor of $1/\prod_{l' = l}^{L-1} s_{l'}^2$. This is in contrast to the factor of $1/\prod_{l' = l}^{L-1} s_{l'}$ for term $\mathcal{C}_l$, which represents the part of the hierarchy where noises were injected all the way down at the edge devices in $\mathcal{S}_L$. This noise reduction of an additional factor of $\prod_{l' = l}^{L-1} s_{l'}$ underscores the rationale behind {\tt M$^2$-FDP}'s integration of MDP with MFL, allowing for an effective reduction in the requisite DP noise for preserving a given privacy level.}

\textcolor{black}{The term $\mathcal{B}_l$ in the bound represents the part of the hierarchy where the noises are injected at some layer between the edge devices and the local aggregation layer $l$. Index $m$ captures the case where noise has been injected at layer $m-1$. Observe that the factor $1/\prod_{l' = l}^{m-1} s_{l'}\prod_{l'' = m}^{ L-1} s_{l''}^2$ can be split into two parts: (i) the effect of all layers below $m$, having the same factor $\prod_{l'' = m}^{ L-1} s_{l''}^2 $ as in $\mathcal{A}_l$; and (ii) the effect of all layers above $m$, have the same factor $\prod_{l' = l}^{ m-1} s_{l'}$ as in $\mathcal{C}_l$.} Thus, overall, we see that \textit{in the trade-off between privacy protection and training accuracy for an HFL system, noises injected at lower hierarchy layers inflict a proportionally larger harm to the accuracy.}

\textcolor{black}{\textbf{Special cases of DP-infused FL:} Theorem~\ref{thm:noncvx} can be used to analyze existing DP-infused FL baselines as special cases: setting $L=1$ and $p_0=0$ reduces {\tt M$^2$FDP} to LDP~\cite{Wainakh2020HLDP}, while $L=2$ and $p_0=p_1=0$ reduces it to HDP~\cite{chandrasekaran2024hierarchical}. In both cases, the key difference manifests in term (b), which becomes $8 M (K^{\max})^4 q^2\log(1/\delta)\alpha_1^2 / \epsilon^2$ for LDP and $16 L M (K^{\max})^4 q^2\log(1/\delta)\alpha_2^2 / (\epsilon^2s_1)$ for HDP; notably, without the sharper decrease according to products of $s_{l}$ and $s_{l}^2$ across layers $l$. By contrast, {\tt M$^2$FDP} achieves the $\mathcal{O}(1/\sqrt{T})$ rate of the FedAvg-style methods, but with a strictly smaller DP-induced error gap, manifesting from noise injection according to heterogeneous trust models. Our experimental results in Sec.~\ref{sec:experiments} will confirm this advantage, showing {\tt M$^2$FDP} consistently reaches target accuracies in fewer rounds under the same privacy budget.}

\textcolor{black}{\textbf{Recovering FL and HFL:} If all servers are secure ($p_l^{\min}=p_l^{\max}=1$ for all $l$), then no DP noise is required at any aggregation step of {\tt M$^2$FDP}, and the constant term (b) in Theorem~\ref{thm:noncvx} disappears completely. Under this condition, our algorithm reduces to Hierarchical FedAvg (HFL). In the special case $L=1$ with a secure server, our results recover conventional FedAvg (FL) without any constant error term. On the other hand, when $L=1$ and the server is insecure, the edge devices must inject DP noise. The effect is captured by the $\mathcal{C}_1$ term in (b) of Theorem~\ref{thm:noncvx}, whereas $\mathcal{A}_1$ and $\mathcal{B}_1$ vanish, reflecting the fact that the device-side noise injection creates the residual constant. This behavior aligns with prior analysis of single-layer DP-FL (e.g., the $\phi$ term in the analysis of~\cite{Shen2022imp}).}

The following corollary of Theorem~\ref{thm:noncvx} more intuitively captures the effect of trustworthiness up to a particular layer:

\begin{corollary} \label{cor:noncvx}
    Under Assumptions~\ref{assump:SGD_noise} and \ref{assump:genLoss}, if $\eta^t=\frac{\gamma}{\sqrt{t+1}}$ with $\gamma\leq\min\{\frac{1}{K^{\max}},\frac{1}{T}\}/\beta$, letting $m$ be the lowest layer with an insecure intermediate node (i.e., $ p_{l'}^{\max} = p_{l'}^{\min} = 1 \; \forall l' \in [m+1 , L]$), then the cumulative average global gradient satisfies
    
    \vspace{-3mm}
{
    \begin{align} \label{eq:cor_noncvx}
        &\textstyle\frac{1}{T}\sum_{t=1}^T \mathbb{E}\left\|\nabla F(w^{(t)})\right\|^2 \leq \textstyle\frac{2\beta \mathbb{E} F(w^{(1)})}{\sqrt{T+1}}+ \frac{K^{\max}\left(G^2\left(1 + \frac{1}{\beta}\right) + \sigma^2\right)}{T}\notag\\
        &\textstyle+ \frac{8 L M (K^{\max})^4 q^2\log(\frac{1}{\delta})}{\epsilon^2}\sum_{l=1}^{m}  \frac{(1 - p_{l-1}^{\min})^2p_l^{\max}\alpha_l^2}{\prod_{l'=l}^{L-1}s_{l'}^2} \\
        &\textstyle+ \underbrace{\textstyle\frac{8 L M (K^{\max})^4 q^2\log(\frac{1}{\delta})}{\epsilon^2}\sum_{l=1}^{m} \frac{(1 - p_{l-1}^{\min})^2(1 - p_l^{\min})(\alpha_l')^2}{\prod_{l'=l}^{m}s_{l'}\prod_{l''=m+1}^{L-1}s_{l''}^2}}_\textrm{(c)}.\notag
\end{align}
}
\vspace{-4mm}
\end{corollary}
 The magnitude of term (c) is increasing in the value of $m$, i.e., as the first layer of insecurity becomes closer to the edge. The gap when $m=m'$ is $s_{m'}$ times larger than the gap when $m=m'-1$.
 This shows how ensuring trustworthiness up to a particular layer gives a benefit of improved training performance for a given DP guarantee.

 {\color{black} More generally, to the best of our knowledge, this is the first work in multi-tier DP-FL to establish a formal convergence guarantee. Prior efforts on the hierarchical setup~\cite{Shi2021HDP,chandrasekaran2024hierarchical} focus on heuristic noise mechanisms and lack convergence analysis. Existing single-tier DP-FL results~\cite{xiong2021privacy} cannot be directly extended to our setup either due to the compounding of noise across layers and the coordination between secure and insecure nodes. Our framework addresses this by showing how trust heterogeneity across tiers directly shapes the error bound. Specifically, the constant gap $(b)$ in Theorem~\ref{thm:noncvx} depends on the proportion of trustworthy versus untrustworthy nodes: updates routed through secure servers fall into $\mathcal{A}_l$ terms, which decrease the error according to $s_{l'}^2$, while those from insecure devices fall into $\mathcal{C}_l$ terms, which decrease less rapidly in $s_{l'}$. Intermediate cases ($\mathcal{B}_l$) interpolate these effects.}

{\color{black} Finally, we remark that with the implementation of {\tt M$^2$FDP}, the DP guarantees in Proposition~\ref{prop:GM} hold true for all nodes in the network. We provide a formal privacy analysis demonstrating this in Appendix~\ref{app:privacy_analysis}.}



\vspace{-0.10in}
\subsection{Proof Sketch and Key Intermediate Results} \label{ssec:sketch}
Here, we provide a proof sketch for Theorem~\ref{thm:noncvx} and Corollary~\ref{cor:noncvx} in terms of key intermediate results. The detailed proofs, including of the intermediate lemmas, can be found in the supplementary material.

First, note that for a given aggregation interval $K^t$, the global average of the local models~\eqref{eq:glob_aggr_3} can be rewritten according to the following dynamics:
{
\begin{align}
\label{eq:skprof_iter}
    &w^{(t+1)} = \sum_{c = 1}^{N_1} \rho_c \left(w_{1,c}^{(t,K^t+1)} + n_{1,c}^{(t,K^t+1)}\right) \nonumber \\
    &= w^{(t)} + n_{DP}^t \notag\\
    &-\eta^t \sum_{k=1}^{K^t} \sum_{c_1 = 1}^{N_1} \rho_{c_1}\ldots\sum_{j\in\mathcal{S}_{L-1,c_{L-1}}} \rho_{L-1, c_{L-1},j} {g}_j^{(t,k)},
\end{align}
}
where 
\begin{align}
    &n_{DP}^t = \sum_{c=1}^{N_1}\rho_c n_{1,c}^{(t,K^t+1)} \notag\\
    &+\sum_{l=1}^{L-1} \sum_{\substack{k \in K_l^t }}\left(\sum_{c_1 = 1}^{N_1} \rho_{c_1}\ldots\sum_{c _l\in \mathcal{S}_l \cap \mathcal{N}_{U,l}} \sum_{i\in \mathcal{S}_{l,c_l}} \rho_{l,c_l,i}n_{l+1,i}^{(t,k)}\right).
\end{align}

{\color{black} By applying $\beta$-smoothness of the global function $F$ onto ~\eqref{eq:skprof_iter}, employing properties of stochastic gradients, and conducting a series of algebraic manipulations, we arrive at \eqref{eq:ld_ncx}. Terms $(a)$ and $(b)$ in \eqref{eq:ld_ncx} can be bounded using Assumption~\ref{assump:genLoss}, while term $(c)$ requires a upper bound on the total DP noise injected at global iteration $t$. Term (c) can be viewed as a weighted average of noises generated across layers. Referring to Lemma 1, we then separate all DP noise effects at network layer $l$ into three cases: (i) noises injected at an intermediate layer $l$, (ii) noises injected at an intermediate layer $l' \in [l+1, L-1]$, (iii) noises injected at the edge device layer ($l = L)$. Based on Lemma~\ref{lem:DeltaM}, we can see that noises in group (i) scale with $1/\prod_{l' = l}^{L-1} s_{l'}^2$, those in group (ii) scale with $1/\prod_{l' = l}^{m-1} s_{l'}\prod_{l'' = m}^{ L-1} s_{l''}^2$, and group (iii) scales with $1/\prod_{l' = l}^{L-1} s_{l'}$. The noise additions from these groups are captured in $\mathcal{A}_l, \mathcal{B}_l, \mathcal{C}_l$ of \eqref{eq:ABC_terms}, and lead to the following lemma:}

\begin{lemma}\label{lem:bd_on_DP_noise_main}
For any given $k \in [1, K^t]$, and any layer $l \in [1, L-1]$, the DP noise that is passed up the $l^{\textrm{th}}$ layer from lower layers can be bounded by:
\begin{align}
   &\mathbb{E}\left\|\sum_{i\in \mathcal{S}_{l,d_{l,i}}} \rho_{l,d_{l,i},i}n_{l+1,i}^{(t,k)} \right\|^2 \notag \\
   &\leq \frac{2\eta^t M (K^t)^2 T q^2G\log(1/\delta)}{\epsilon^2} \left(\mathcal{A}_l + \mathcal{B}_l + \mathcal{C}_l\right),
\end{align}
where $\mathcal{A}_l$, $\mathcal{B}_l$, and $\mathcal{C}_l$ are defined at \eqref{eq:ABC_terms}.
\end{lemma}

Integrating Proposition~\ref{prop:GM}, Lemma~\ref{lem:Delta}, and Lemma~\ref{lem:bd_on_DP_noise_main} into~\eqref{eq:ld_ncx}, and considering that $\eta^t \leq \frac{1}{\max \{T, K^t\}\beta}$, we obtain the following expression:

{\small
\begin{align}
    &\frac{\eta^t K^t}{2}\mathbb{E}\left\|\nabla F(w^{(t)})\right\|^2 \leq \mathbb{E}\left[ F(w^{(t)}) -  F(w^{(t+1)})\right] \notag\\
    &+ \frac{(\eta^2K^t)^2 \beta}{2} \left(G^2 + \sigma^2 +  \frac{G^2}{K^t\beta}\right)\notag\\
    &+\frac{4(\eta^t K^t)^2 L\beta  M (K^t)^2 T q^2\log(\frac{1}{\delta})}{\epsilon^2}\sum_{l=1}^{L}(1 - p_{l-1}^{\min})^2 \bigg(\mathcal{A}_l + \mathcal{B}_l + \mathcal{C}_l\bigg).
\end{align}
}
Further algebraic manipulation yields the final results.

\section{Adaptive Control Algorithm} \label{sec:ctrl_DPFL}
\noindent 
We now show how our convergence analysis can be employed to develop control algorithms that optimize over the privacy-utility tradeoff.
We consider three key control variables: 
\begin{enumerate}[leftmargin=8mm]
    \item The size of the gradient descent step \{$\eta^t$\} across global iterations $t$.
    \item The local training interval length $K^t$ in-between global aggregations $t$ and $t+1$.
    \item A partially sampled set of edge devices $\mathcal{S}_{L-1,c}^t \subseteq \mathcal{S}_{L-1,c}$ that will participate in round $t$ for each intermediate node $c \in \mathcal{S}_{L-1}$. We specifically aim to control the size $s_c^t$ of the participating sample set:\footnote{It is worth mentioning that as long as the sampling process is unbiased (e.g., uniform random), the theoretical guarantees from Theorem~\ref{thm:noncvx} hold under partial participation.}
    \begin{equation}
    s_c^t \overset{\Delta}{=} |\mathcal{S}_{L-1,c}^t|.
\end{equation}
\end{enumerate}

Our control algorithm framework is summarized in Fig.~\ref{fig:ctrl}. 
The central server orchestrates the adjustable parameters during the global aggregation step. To simplify the presentation of our algorithm, we assume $\mathcal{K}_l^t = \emptyset$ for $l \geq 2$, which means that local aggregations are occurring at layer $l = 1$.
We also assume $p_l^{\min} = p_l^{\max}$, which means all insecure aggregations between layers $l$ and $l-1$ have the same ratio of children with and without DP noise injection. We assume the network operator will specify (i) desired $(\epsilon,\delta)$ privacy requirements, (ii) a target number of local model updates $\tau = \sum_{t=1}^{T} K^t$, and (iii) a maximum duration between global aggregations $K^\textrm{max}$.

Our control algorithm has two parts: \textit{Part I} employs an adaptive approach (detailed in Sec.~\ref{subsec:learniParam}) for calibrating the step-size to ensure the convergence performance from Theorem~\ref{thm:noncvx}. \textit{Part II} adopts an optimization framework (outlined in Sec.~\ref{subsec:learnduration}) to adapt $s_c^t$ and $K^t$, balancing the objectives of ML loss and resource consumption for the target DP requirement.

\vspace{-0.10in}
\subsection{Part I: Step Size Parameter ($\gamma^t$)}\label{subsec:learniParam}
\textcolor{black}{After global round $t-1$ ends, we begin by fine-tuning the step size parameter $\gamma^t$ for round $t$, taking into account the relevant measures smoothness  value $\beta$, total gradient computations $\tau$, and local iterations $K$.} The server is tasked with estimating $\beta$, for which we follow Section IV-C of~\cite{lin2021timescale}. $\tau$ is pre-specified, where $\tau$ is always strictly larger than the total number of global aggregations $\tau \geq T$. $K$ is either taken as $K^{t-1}$ from the previous interval, or initialized for the first interval $K^1$. Given that higher feasible $\gamma^t$ values enhance the step sizes, leading to faster model convergence as per the conditions described in Theorem~\ref{thm:noncvx}, we identify the maximum $\gamma^t$ value that complies with $\gamma^t \leq \min\{\frac{1}{K^{t-1}},\frac{1}{\tau}\} / \beta$.
  
\begin{figure}[t]
\includegraphics[width=0.47\textwidth]{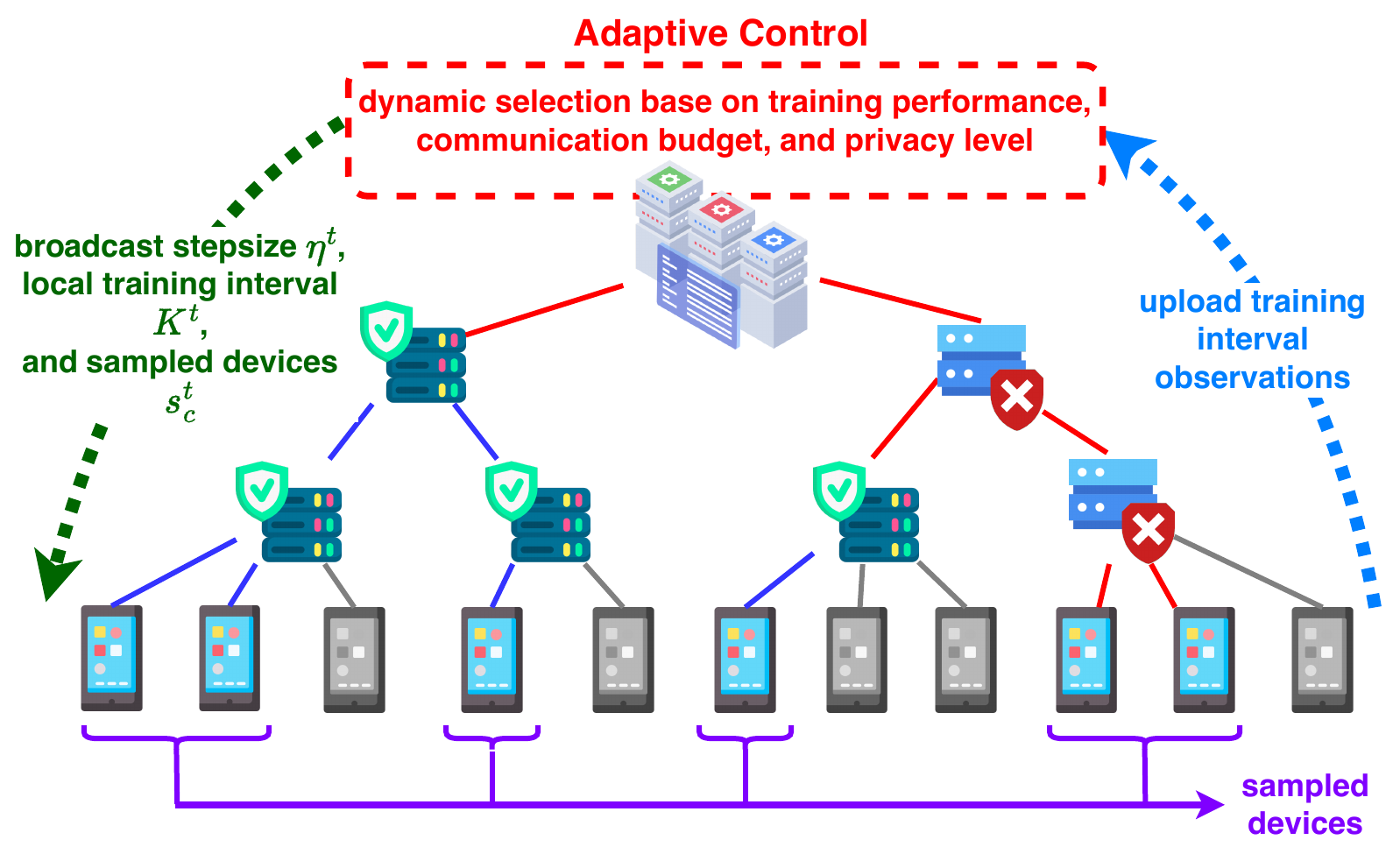}
\centering
\caption{Overview of the adaptive control algorithm, outlining its objectives, adjustable parameters, and observations.}
\label{fig:ctrl}
\end{figure} 

 \vspace{-0.10in}
\subsection{Part II: Training Interval ($K^t$) and Participation ($s_{c}^t$)}\label{subsec:learnduration}
\textcolor{black}{We proceed to craft an optimization problem that determines the number of local iterations $K^t$ and the number of edge devices $\{s_{c}^t\}_{c\in \mathcal{S}_{L-1}}$ sampled in each cluster $c$ at layer $L-1$ for global round $t$.}
\eqref{eq:adaptive_alg} is designed to jointly optimize four competing objectives: (O1) the energy consumption associated with global model aggregations, (O2) the communication delays incurred during these aggregations, and (O3) the performance of the global model, taking into account the impact of the DP noise injection procedure in {\color{black} {\tt M$^2$FDP}} dictated by Theorem~\ref{thm:noncvx}. Formally, we have:
\begin{align}
\label{eq:adaptive_alg}
    \min_{K^t, \{s_{c}^t\}_{c\in \mathcal{S}_{L-1}}} &  \alpha_1 \underbrace{E(K^t, \{s_c^t\})}_\text{$(a)$}+\alpha_2 \underbrace{\Gamma(K^t, \{s_c^t\})}_\text{$(b)$}\notag\\
    &+\alpha_3 \underbrace{\nu (K^t, \{s_{c}^t\})}_\text{$(c)$} \notag\\
    \textrm{subject to} 
   & \;\;\; 1 \leq K^t \leq \min{\{K^\textrm{max}, \tau^t\}}, K^t\in\mathbb{Z}^+, \notag\\  
   & \;\;\; 1\leq s_{c}^t \leq |\mathcal{S}_{L-1, c}|,~s_{c}^t \in \mathbb Z^+,
   \tag{$\mathcal{P}$}
\end{align}

\noindent where $\tau^t = \tau - \sum_{t'=1}^{t-1}K^{t'}$ is the remaining number of local model updates.

\textbf{Objectives}: Term $(a)$ captures the communication and computation energy expenditure over the estimated remaining global aggregations, where
\begin{align*}
&E(K^t, \{s_c^t\}) = \frac{\tau^t}{K^t}E_{\textrm{Glob}}( \{s_c^t\})\\& +\frac{\tau^t}{K^t}\left(|\mathcal{K}_1|\sum_{c'\in \mathcal{S}_{1}} E_{\textrm{c',Loc}}( \{s_{c}^t\}) + K^t\sum_{c''\in \mathcal{S}_{L-1}}s_{c''}^t{E_{iter}}\right).
\end{align*}
The energy consumption from local model aggregation at an intermediate node, denoted as $E_{\textrm{c,Loc}}$, is calculated by summing up the energy used for communication between edge devices and intermediate node $c$. Similarly, the global aggregation energy consumption, $E_{\textrm{Glob}}$, accumulates the energy used for communications during global aggregation. $E_{iter}$ is the total energy for one edge device to compute one local update, which we assume is constant (e.g., uniform mini-batch sizes).

Term $(b)$ captures communication and computation delay incurred over the estimated remaining local intervals:
\begin{align*}
    &\Gamma(K^t, \{s_c^t\})\\
    &=\frac{\tau^t}{K^t}\left(\Gamma_{\textrm{Glob}}( \{s_c^t\}) + |\mathcal{K}_1|\sum_{c'\in \mathcal{S}_1}\Gamma_{\textrm{c',Loc}}(s_{c'}^t)+ K^t\Gamma_{iter}\right).
\end{align*}
The delay for local aggregation, $\Gamma_{\textrm{c,Loc}}$, captures the total consumed time it takes for all selected devices to transmit the model updates to intermediate node $c$ based on their transmission rates. The global aggregation delay, denoted as $\Gamma_{\textrm{Glob}}$, represents the communication time to perform one round of global aggregation. $\Gamma_{iter}$ is the delay for the edge devices to perform one round of local update, again constant.


Finally, term $(c)$ represents the upper bound on the optimality gap, quantified as term (b) in Theorem~\ref{thm:noncvx}:
\begin{align}
    &\nu(K^t, \{s_c^t\}) \overset{\Delta}{=} \frac{8LM(K^t)^4 q^2 \log(\frac{1}{\delta})}{\epsilon^2} \notag\\
&\cdot\sum_{l=1}^{L}(1 - p_{l-1})^2 \bigg(\mathcal{A}_l + \mathcal{B}_l + \mathcal{C}_l\bigg).
\end{align}
A lower value is thus in line with better ML performance.





\textbf{Constraints}: The first constraint ensures that the value of $K^t$ remains within a preset range, i.e., to prevent any given local training interval from becoming too long. Meanwhile, the second constraint ensures that the number of participating devices in each cluster does not exceed the subnet size. Increasing $s_{1,c}^t$ will hinder the energy objectives, while from Theorem~\ref{thm:noncvx}, we see that it will improve the stationarity gap objective term (d). The value of $K^t$ has the opposite effect, as increasing it causes the stationarity gap to grow, but simultaneously reduces the frequency of aggregations. 

\textbf{Solution}: \eqref{eq:adaptive_alg} is classified as a non-convex mixed-integer programming problem due to term $(c)$. The total time complexity required if using a nested line search strategy is $\mathcal O(N_{L-1}\times K^{max}\times s^{\max})$. To mitigate this, instead of searching an optimal set of participating devices for all subnets, we decrease the decision variable space from $N_{L-1}$ to 2, by only differentiating sample rates between subnets $\mathcal{S}_{L-1, c}$ which have secure versus insecure ancestor nodes at layer $l=1$.

\vspace{-1mm}   
\section{Experimental Evaluation}
\label{sec:experiments}



\subsection{Simulation Setup}
\label{ssec:setup} 
By default, we consider a hierarchical FL network with $L = 2$, comprising a client layer, an edge layer, and a cloud layer. In total, $N_2 = 50$ edge devices are evenly distributed across $N_1 = 10$ edge subnetworks (subnets). In Sec. \ref{sec:more_layers}, we also consider multi-tier FL settings with more network layers.

We use four datasets commonly employed for image classification tasks: EMNIST-Letters (E-MNIST)\cite{cohen2017emnist}, Fashion-MNIST (F-MNIST)\cite{xiao2017}, CIFAR-10\cite{krizhevsky2009learning}, and CIFAR-100\cite{krizhevsky2009learning}. 
Following prior work~\cite{wang2019adaptive,lin2021timescale}, the training samples from each dataset are distributed across the edge devices in a non-i.i.d manner. 
The first two datasets employ Support Vector Machine (SVM), while the latter two datasets employ a 12-layer convolutional neural network (CNN) accompanied by softmax and cross-entropy loss.
The model dimensions are set to $M=7840$. Comparison between {\color{black} SVM} and CNN provides insight into {\tt M$^2$FDP}'s performance when handling convex versus non-convex loss functions.
Also, unless otherwise stated, we assume $p_c = 0.5$, and that semi-honest entities are all governed by the same DP budget $\epsilon=1, \delta=10^{-5}$.

For the control algorithm, the energy and delay model for our system accounts for various factors including model size \(M\), quantization level \(Q\), individual device transmit power \(p_j\), and transmission rates \(R_j^{(t)}\), as in in~\cite{lin2021timescale}. 
Wireless communications between each network layer below the cloud server employ a standard transmit power of \(p_{l,i} = 24\) dBm per device, with a data rate of \(\bar{R}_{l,i}^{(t)} = 35\) Mbps. For wired communications between intermediate nodes and the cloud server, the transmit power is set at \(\bar{p}_{c} = 38\) dBm with a data rate of \(\bar{R}_{c}^{(t)} = 100\) Mbps. The delay of local iterative update is set to $\Gamma_{iter} = 2\times 10^{-1}$ secs, and the energy cost for local updates is set to $E_{iter} = 1 \times 10^{-3}J$.

\begin{figure*}[t]
\includegraphics[width=1.0\textwidth]{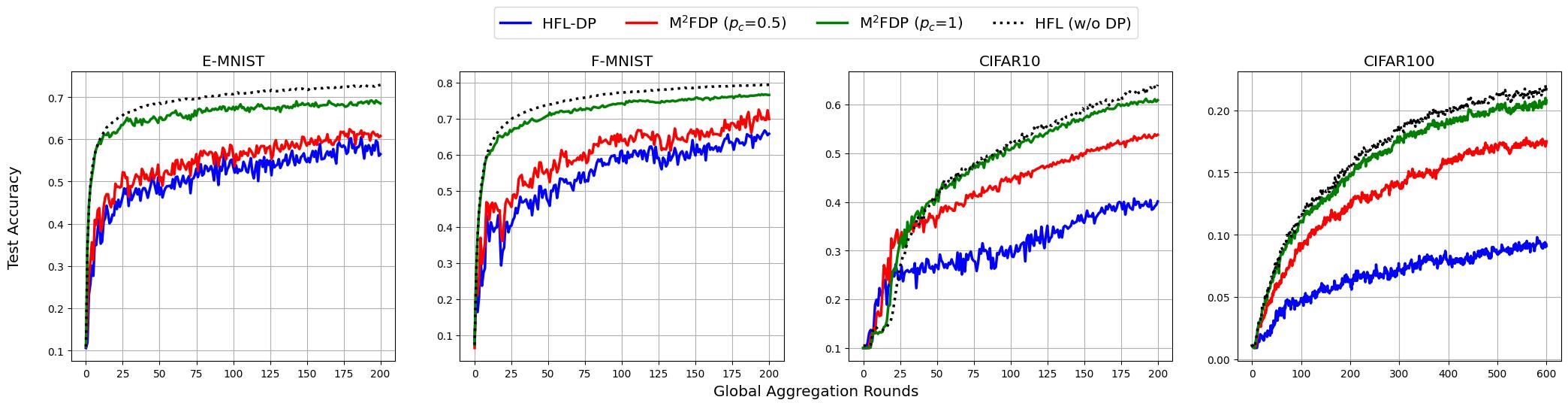}
\centering
\caption{Performance comparison between {\tt M$^2$FDP}, the {\tt HFL-DP} baseline from~\cite{Shi2021HDP}, and an upper bound established by hierarchical {\tt FedAvg} without DP. {\tt M$^2$FDP} significantly outperforms {\tt HFL-DP} and is able to leverage trusted edge servers effectively.}
\label{fig:mnist_poc_1_all}
\vspace{-0.2in}
\end{figure*} 

\vspace{-0.2mm}
\subsection{{\tt M$^2$FDP} Comparison to Baselines}
\label{ssec:conv-eval}

 
 

\begin{figure}[t]
\includegraphics[width=0.49\textwidth,height=0.15\textheight]{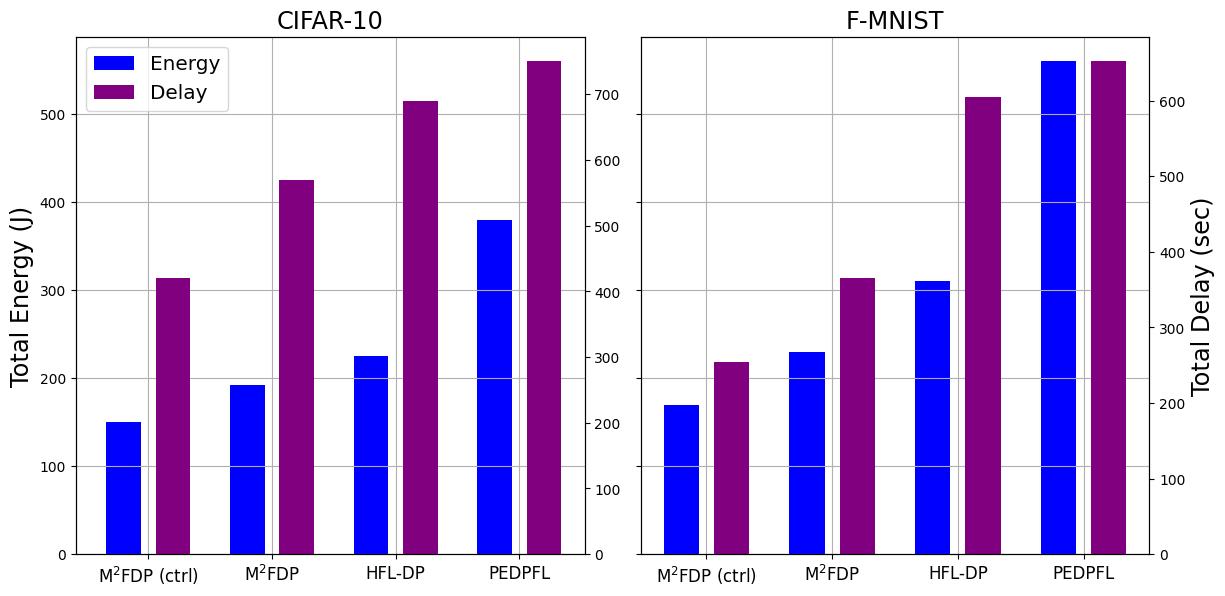}
\centering
\vspace{-0.2in}
\caption{Comparison of {\tt M$^2$FDP} with adaptive parameter control to the baselines in total energy and delay upon reaching 75\% testing accuracy for F-MNIST and 45\% testing accuracy for CIFAR-10. {\tt M$^2$FDP} obtains substantial improvements in both metrics for both F-MNIST and CIFAR-10.}
\label{fig:res}
\vspace{-0.1in}
\end{figure}
\begin{figure}
\includegraphics[width=0.49\textwidth]{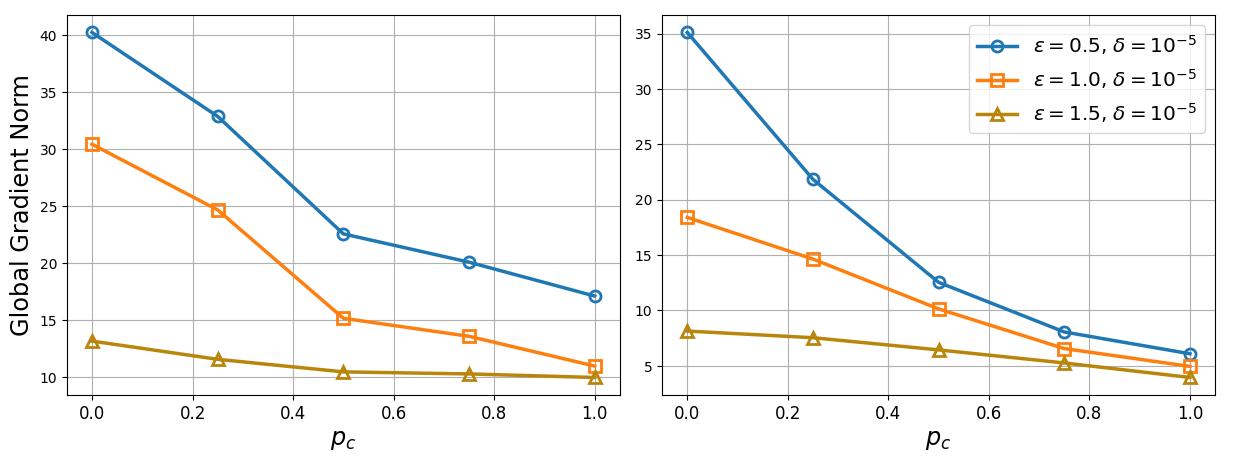}
\includegraphics[width=0.49\textwidth]{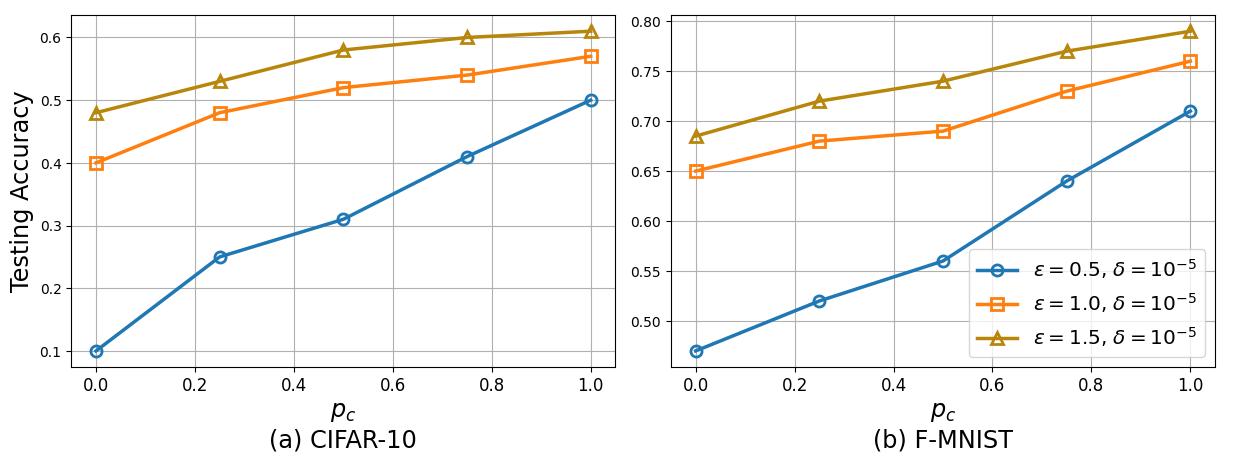}
\centering
\vspace{-0.2in}
\caption{\textcolor{black}{Interplay between privacy and performance in {\tt M$^2$FDP} across various probabilities ($p_c$) of a subnet's linkage to a secure edge server under different privacy budgets ($\epsilon$) after 200 aggregation rounds.} \vspace{0.05in}} 
\label{fig:eps_vs} 
\end{figure}

\begin{figure}[t]
\includegraphics[width=0.49\textwidth]{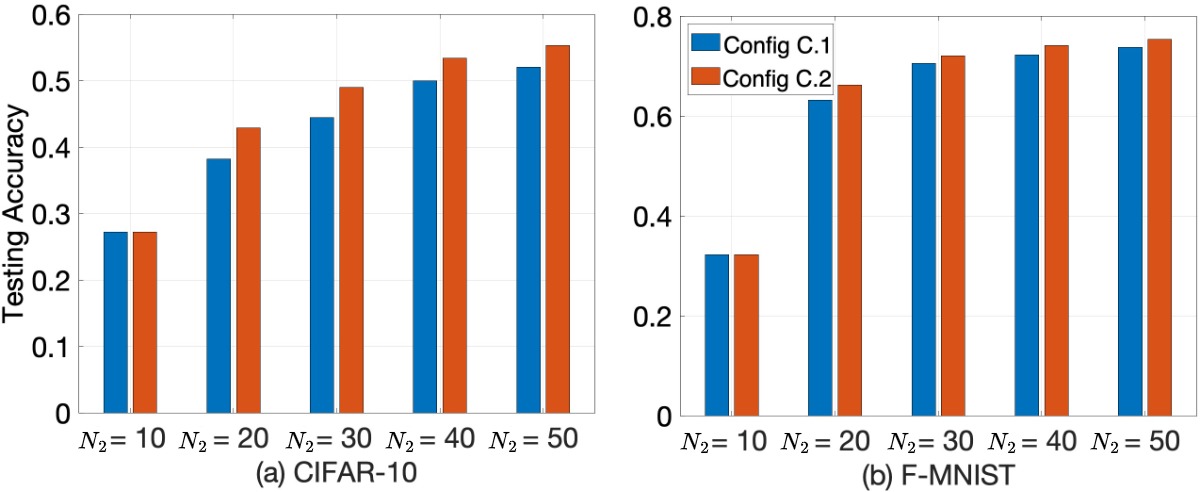} 
\centering
\vspace{-0.2in}
\caption{Impact of various network configurations on the performance of {\tt M$^2$FDP}. Under the same network size, enhancing the size of each subnet $s_c$ yields superior test accuracy compared to merely increasing the number of subnets $N$.}
\label{fig:mnist_poc_2_all}
\end{figure}
Despite the increasing interest in HFL, there remains a significant lack of research addressing differential privacy in this context. As a result, the number of established baseline algorithms available for comparison is limited.
Our first experiments examine the performance of {\tt M$^2$FDP} compared with the following baselines: We utilize the conventional hierarchical {\tt FedAvg} algorithm~\cite{liu2020client}, which offers no explicit privacy protection, as our upper bound on achievable accuracy (labeled {\tt HFL (w/o DP)}). We also implement {\tt HFL-DP}~\cite{Shi2021HDP}, which employs LDP within the hierarchical structure, for competitive analysis. Further, we consider {\tt PEDPFL}~\cite{Shen2022imp}, a DP-enhanced FL approach developed for the standard star-topology structure, for which we assume the edge devices all form a single cluster and apply LDP.



\subsubsection{Training Convergence Performance}
Fig.~\ref{fig:mnist_poc_1_all} demonstrates results comparing {\tt M$^2$FDP} to the baselines. Each algorithm employs a local model training interval of $K = 20$ and conduct local aggregations after every five local SGD iterations. We see {\tt M$^2$FDP} obtains performance enhancements over {\tt HFL-DP} by exploiting secure intermediate nodes in the hierarchical architecture. This improvement is observed both in terms of superior accuracy and decreased accuracy perturbation as the ratio of secure nodes ($p_1$) increases. 
Specifically, when $p_1 = 0.5$, {\tt M$^2$FDP} achieves an accuracy gain at $t = 200$ of $5\%$ for E-MNIST, $5\%$ for F-MNIST, $13\%$ for CIFAR-10, and at $t=600$ a gain of $16\%$ for CIFAR-100, respectively, and displays reduced volatility in the accuracy curve compared to {\tt HFL-DP}. When all edge servers are secure ($p_1=1$), the improvement almost doubles in E-MNIST and F-MNIST, while the improvement decreases for CIFAR-10 and CIFAR-100. 
Notably, compared to the upper bound benchmark, {\tt M$^2$FDP} with $p_1=1$ achieves an accuracy within $4\%$ of the benchmarks of all four datasets.
In other words, the exploitation of secure nodes in the middle layers of the hierarchy significantly mitigates the amount of noise required to maintain a desired privacy level.


\subsubsection{Adaptive Control Algorithm Performance}
\label{ssec:control-eval}
Next, we evaluate {\tt M$^2$FDP}'s control algorithm performance. In this setting, all experiments started with $K^t = 5$ and edge device sample rate $100\%$. Fig.~\ref{fig:res} demonstrates the results,
where the total energy consumption (O1) and total delay (O2) are assessed upon the global model achieving a testing accuracy of 75\%.
Overall, we see that {\tt M$^2$FDP} with control ({\tt ctrl}) substantially improves over the baselines for both metrics. For (O1), the blue bars show reductions in energy consumption by $17.6\%$ and $25.1\%$ compared to {\tt M$^2$FDP} without control, by $30.3\%$ and $46.2\%$ compared to {\tt HFL-DP}, and by $62.6\%$ and $71.5\%$ compared to {\tt PEDPFL} for the CIFAR-10 and FMNIST datasets, respectively. Similarly, for (O2), the red bars indicate that {\tt {\tt M$^2$FDP}(ctrl)}'s delay is $18.1\%$ and $24.9\%$ less than {\tt M$^2$FDP} without control, $30.1\%$ and $43.2\%$ less than {\tt HFL-DP}, and $36.2\%$ and $51.3\%$ less than {\tt PEDPFL} for the CIFAR-10 and FMNIST datasets. These results highlight the enhanced performance and efficiency in resource usage offered by {\tt M$^2$FDP} through its adaptive parameter control, which optimizes the balance between the optimality gap (as established in Theorem~\ref{thm:noncvx}), communication delay, and energy consumption. Notably, the improvement in both metrics underscores the advantage of the joint device participation and training interval optimization strategy employed by {\tt M$^2$FDP}.
 


\vspace{-0.2mm}
\subsection{Impact of System Parameters}

\subsubsection{Privacy Budget and Secure Nodes} \label{subsubsec:combiner}
\textcolor{black}{We next examine the impact of the ratio $p_c$ under different $(\epsilon,\delta)$-DP guarantees. The results are shown in Fig.~\ref{fig:eps_vs}, where $p_c=0$ corresponds to {\tt HFL-DP}. The upper plots of Fig.~\ref{fig:eps_vs} depict the gradient norm of the global model after training. We observe that the gradient norm increases as the privacy requirement becomes more stringent, and also rises as $p_c$ decreases (i.e., when more nodes are insecure). At $p_c = 1$, the gradient norm is smallest, but non-zero, as the main server is still insecure. These trends are consistent with our theoretical analysis of {\tt M$^2$FDP}, including the DP-induced term (b) in Theorem~\ref{thm:noncvx}.}

The lower plots in Fig.~\ref{fig:eps_vs} shows that {\tt M$^2$FDP} obtains a considerable improvement in the privacy performance tradeoff as the probability increases. Specifically, under the same privacy conditions, {\tt M$^2$FDP} exhibits an improvement of at least $20\%$ for CIFAR-10 and $10\%$ for F-MNIST when all edge servers in the mid-layer can be trusted (ie, $p_c=1$) compared to {\tt HFL-DP} (i.e., $p_c = 0$). For example, when $p_c=1$ and $\epsilon=0.5$, {\tt M$^2$FDP} achieves accuracy boosts of $40\%$ for CIFAR-10 and $25\%$ for F-MNIST. 
We see that {\tt M$^2$FDP} can significantly mitigate the accuracy degradation from a more stringent DP requirement when more secure aggregators are available. 

\subsubsection{Varying Subnet Sizes} \label{subsubsec:netSize}

Next, we investigate the impact of different network configurations. Two distinct configurations are evaluated:
\begin{itemize}[leftmargin=8mm]
\item {\tt Config. C.1}, where the size of subnets is kept at $s_1 = 5$ as the number of subnets ($N_1$) varies;

\item {\tt Config. C.2}, where $N_1$ is fixed at $2$ while $s_1$ varies.
\end{itemize}
The results are shown in Fig.~\ref{fig:mnist_poc_2_all}, where a positive correlation between network size and model performance is apparent in both configurations.
Specifically, the accuracy gain can be as substantial as $25\%$ and $42\%$ for CIFAR-10 and F-MNIST, respectively, when the number of edge devices transitions from $N_2=10$ to $N_2=50$.
{\color{black} This observation aligns with Theorem~\ref{thm:noncvx}}, which quantifies the noise reduction as the network size expands.

Also, {\tt Config. C.2} leads to superior model performance compared to {\tt Config. C.1}: increasing subnet sizes is more beneficial than increasing the number of subnets. This pattern once again aligns with Theorem~\ref{thm:noncvx}: when subnets are linked to a secure intermediate node, the extra noise needed to maintain an equivalent privacy level can be downscaled by $1/s_1^2$.
These findings underscore the importance of considering network configuration in privacy-accuracy trade-off optimization for {\tt M$^2$FDP}.

\vspace{-0.10in}
\subsection{Results with More Tiers ($L \geq 3$) }\label{sec:more_layers}
In this subsection, we change our setting from  the $L = 2$ networks to multi-tier networks, and interchanging parameters that are tier-related to observe the effect on performance.

\begin{figure}[t]
\vspace{-0.1in}
\includegraphics[width=0.49\textwidth]{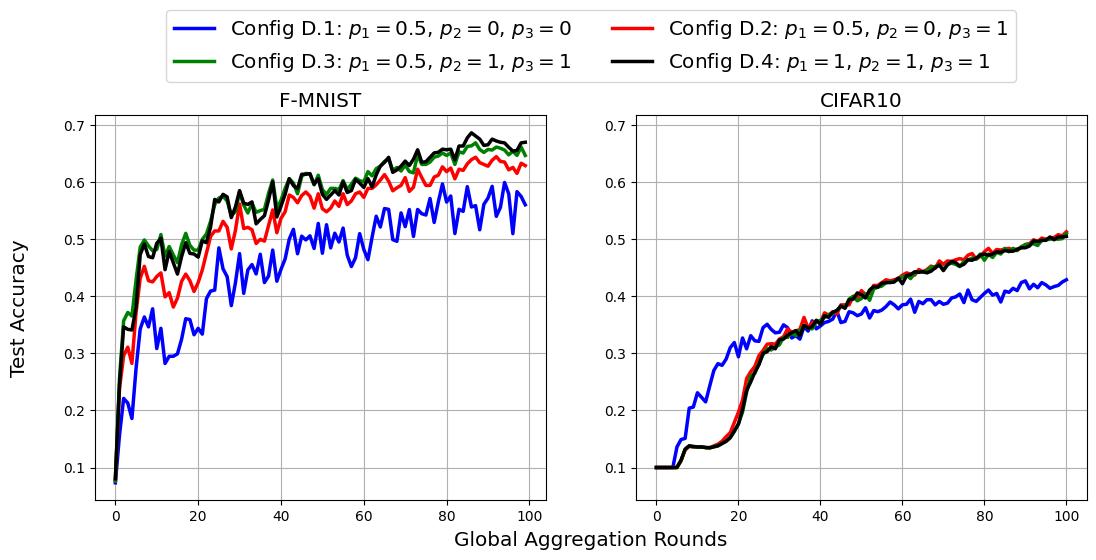} 
\centering
\vspace{-0.2in}
\caption{Impact of changing the ratio of secure intermediate nodes in each layer on F-MNIST and CIFAR-10. The convergence result improved as more layers $l$ have all nodes secure ($p_l = 1$).}
\label{fig:secure_ratio}
\end{figure}

\begin{figure}[t]
\includegraphics[width=0.49\textwidth]{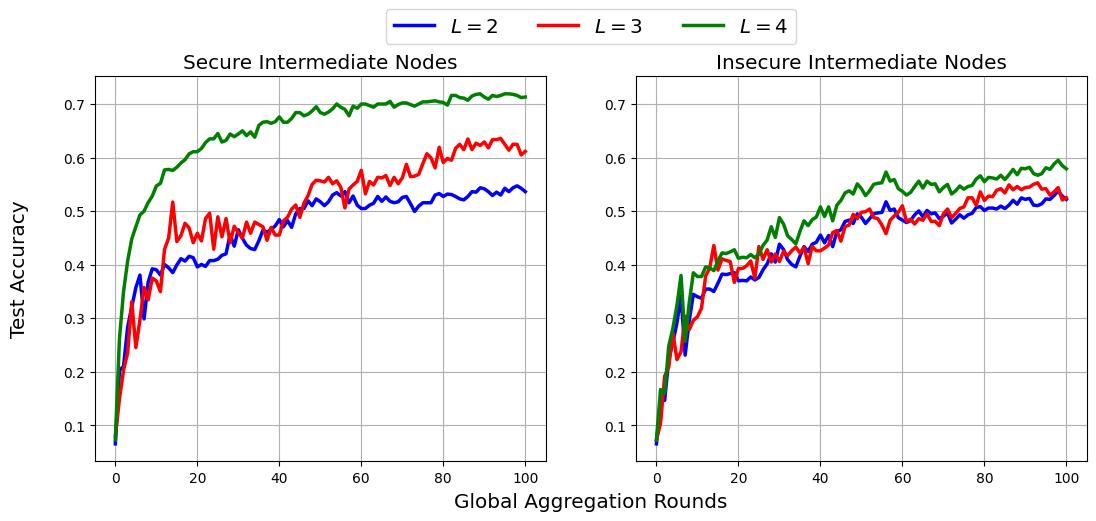} 
\centering
\vspace{-0.2in}
\caption{Impact of changing the number of layers of the multi- network on F-MNIST. The results shows that a larger number of tiers is beneficial to secure intermediate nodes, while unhelpful to insecure intermediate nodes.}
\label{fig:layer_number}
\end{figure}

\subsubsection{The Number of Privacy Protected Layers} We first consider a fixed $L = 4$ multi-tier network, while varying the ratio of secure/insecure intermediate nodes on each tier. There are four configurations considered:
\begin{itemize}[leftmargin=8mm]
\item {\tt Config. D.1}, where $p_1 = 0.5, p_2 = 0, p_3 = 0$, which means half of all intermediate nodes are secure/insecure;
\item {\tt Config. D.2}, where $p_1 = 0.5, p_2 = 0, p_3 = 1$, which means all intermediate nodes in layer $l=3$ are changed to secure nodes, with the remaining nodes the same as the configuration i;
\item {\tt Config. D.3}, where $p_1 = 0.5, p_2 = 1, p_3 = 1$, which means all intermediate nodes in layer $l=2$ are changed to secure nodes, with the remaining nodes the same as the configuration ii;
\item {\tt Config. D.4}, where $p_1 = 1, p_2 = 1, p_3 = 1$, which means all intermediate nodes are secure.
\end{itemize}

Fig. \ref{fig:secure_ratio} demonstrates that {\tt Config. D.2, D.3}, and {\tt D.4} outperform {\tt Config. D.1} on both the F-MNIST and CIFAR-10 datasets. This observation aligns with the result presented in Corollary~\ref{cor:noncvx}, which states that an increased number of layers fully composed of secure intermediate nodes reduces the stationarity gap, thereby enhancing model performance. However, the performance differences among {\tt Config D.2, D.3}, and {\tt D.4} exhibit varying behaviors between F-MNIST and CIFAR-10. This discrepancy can be attributed by magnitudes of certain terms from Theorem~\ref{thm:noncvx}: when terms ($a_1$) and ($a_2$) in Theorem~\ref{thm:noncvx} are relatively larger compared to ($b$), the influence of ($b$) does not dominate the final convergence behavior. As a result, the model performance remains similar even when the proportion of secure intermediate nodes varies.

\subsubsection{Varying Network Tiers} We then investigate how the performance of {\tt M$^2$FDP} changes when the number of tiers changed under certain privacy settings. We consider the following three networks: 1) A three layer network $L=2$, where there are two intermediate nodes at layer 1 $N_1 = 2$, and 8 edge devices at layer 2 $N_2 = 8$. 2) A four layer network $L = 3$, where $N_1 = 2$, $N_2 = 8$, and $N_3 = 32$. 3) A five layer network $L = 4$, where $N_1 = 2$, $N_2 = 8$, $N_3 = 32$, and $N_4 = 128$. In other words, each network configuration is the previous network with four child nodes attached to all leaf nodes. We conduct experiments on two extreme cases, the first case is $p_1 = p_2 = \ldots = p_{L-1} = 1$, which means all intermediate nodes are secure, and the second case is $p_1 = p_2 = \ldots = p_{L-1} = 0$,  which means all intermediate nodes are insecure.

Fig. \ref{fig:layer_number} illustrates the contrasting behaviors of networks with secure and insecure intermediate nodes as the number of tiers increases. For the F-MNIST dataset, networks with only secure intermediate nodes exhibit significant accuracy improvements, achieving a $7\%$ gain when increasing the tiers from $L=2$ to $L=3$, and a $10\%$ gain from $L=3$ to $L=4$. In contrast, networks with only insecure intermediate nodes maintain relatively consistent accuracy across varying tiers. This trend is consistent with the total magnitude of differential privacy (DP) noise, as indicated in Lemma \ref{lem:DeltaM}. Specifically, the variance of DP noise decreases with an increasing number of tiers for networks comprising secure intermediate nodes, whereas it remains unchanged for networks with insecure intermediate nodes.

\begin{figure*}[t]
\includegraphics[width=0.95\textwidth,height=0.15\textheight]{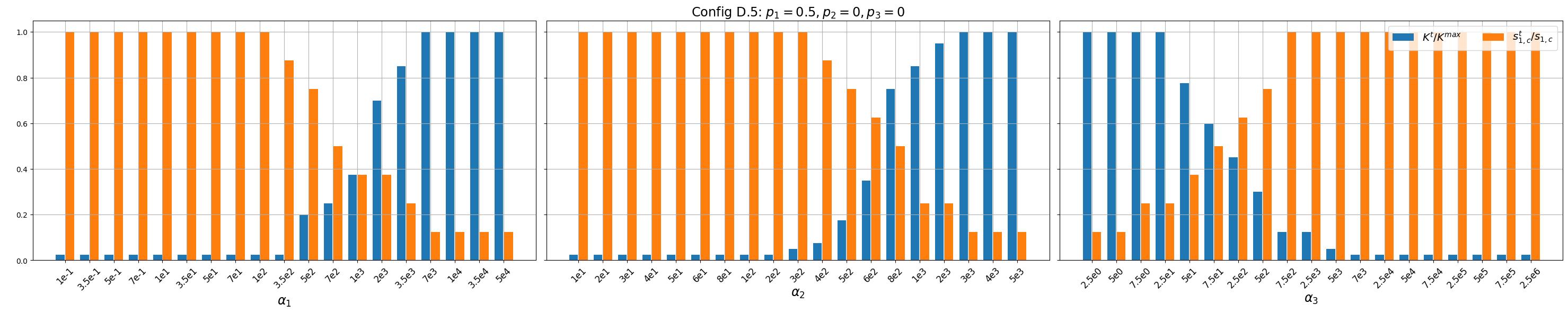}
\includegraphics[width=0.95\textwidth,height=0.15\textheight]{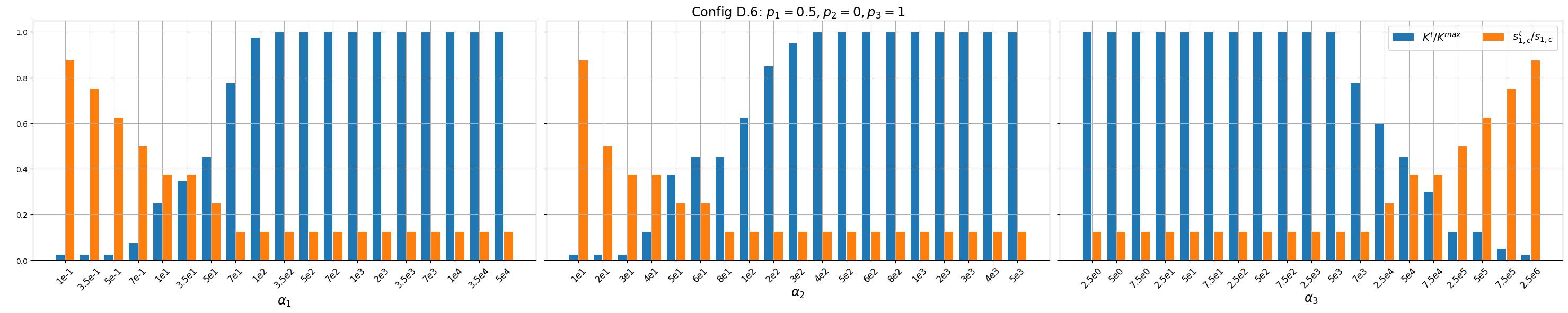}
\centering
\vspace{-0.1in}
\caption{Average values of \(K^t\) and \(s_c^t\) chosen by \eqref{eq:adaptive_alg} across various configurations of coefficients \(\alpha_1\), \(\alpha_2\) and \(\alpha_3\) under two security configurations. One with less secure intermediate nodes (Config 1, $p_1 = 0.5, p_2 = 0, p_3 = 0$), the other with more secure intermediate nodes (Config 2,  $p_1 = 0.5$, $p_2 = 0$, $p_3 = 1$).} 
\label{fig:weight}
\vspace{-5mm}
\end{figure*}

\begin{figure}[t]
\includegraphics[width=0.49\textwidth]{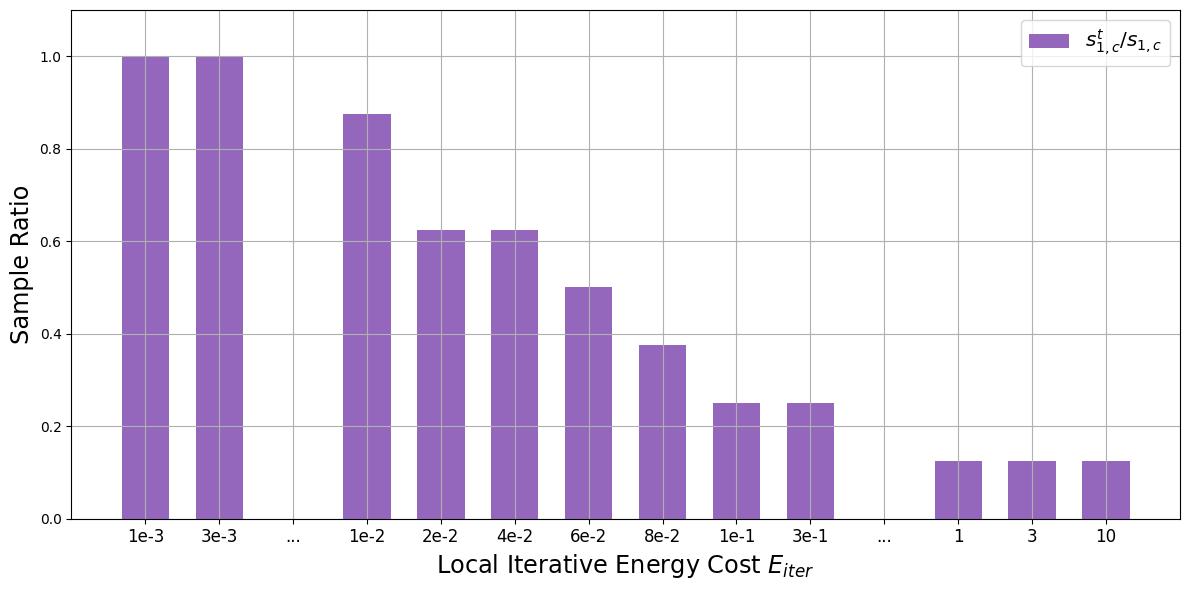} 
\centering
\vspace{-0.2in}
\caption{Average value of sample ratio $s_c^t$ chosen by \eqref{eq:adaptive_alg} under different magnitudes of energy expended per local iteration $E_{iter}$. The control algorithm adapts to a higher energy cost by sampling less edge devices for participation.} 
\label{fig:ctrl_energy}
\end{figure}

\subsubsection{Objective Weights in \eqref{eq:adaptive_alg}}
Fig.~\ref{fig:weight} investigates the control algorithm's response to varying optimization weights in \eqref{eq:adaptive_alg}. In this experiment, we set $L = 4$, $K^{max} = 40$, and $|\mathcal{S}_{L-1, c}| = 8$, and use the F-MNIST dataset. The plotted values of $K^t$ and $s_c^t$ are averaged over the entire training process. We compare the change of the values under two network secure ratios: (i) {\tt Config. D.5}, $p_1 = 0.5, p_2 = 0, p_3 = 0$, which has a lower number of secure intermediate nodes, and (ii) {\tt Config. D.6}, $p_1 = 0.5, p_2 = 0, p_3 = 1$, which has a higher number of secure intermediate nodes.

We see that an increase in \( \alpha_1, \alpha_2 \), the weight on communication energy and communication delay, results in an extended training intervals $K^t$, attributed to less frequent global aggregations cutting down on communications. We also observe a decreased count $s_c^t$ of devices engaged in training, thereby reducing communication overhead.
In contrast, increasing \( \alpha_3 \), the weight on the ML performance, leads to more frequent global aggregations, reducing extensive intervals of local training that could lead to biased local models, and augments the number of devices participating in training, leveraging a larger pool of training data for enhanced learning performance.

In comparing {\tt Config. D.5} and {\tt D.6}, it is evident that {\tt Config. D.6} requires a larger value of $\alpha_3$ to achieve comparable results to {\tt Config. D.5}. This can be attributed to the fact that the network in {\tt Config. D.6} incorporates a greater number of secure intermediate nodes. {\color{black} Based on} Theorem~\ref{thm:noncvx}, the stationarity gap caused by DP noises is significantly smaller in {\tt Config. D.6} than in {\tt Config. D.5}. As a result, a substantially larger $\alpha_3$ coefficient is necessary to yield the same optimal control parameters, $K^t$ and $s_c^t$. This observation also explains the behavior of the $\alpha_1$ and $\alpha_2$ plots, where the plot corresponding to {\tt Config. D.6} exhibits a leftward shift relative to that of {\tt Config. D.5}.

Fig.~\ref{fig:ctrl_energy} illustrates the impact of energy levels on the control algorithm's performance during local iterative updates. When the energy cost per iteration for each edge device $E_{iter}$ is low, the system tends to sample a larger number of devices for training, leveraging the reduced cost of gradient descent steps on edge devices. Conversely, when $E_{iter}$ is high, indicating a greater energy expenditure for each gradient descent step, the system adapts by sampling fewer devices for training. This behavior demonstrates our control algorithm's adaptability to varying system energy constraints, ensuring efficient training under diverse energy settings.

\section{Conclusion and Future Work}
\noindent In this study, we developed {\tt M$^2$FDP}, which integrates multi-tier differential privacy (MDP) into multi-tier federated learning (MFL) to enhance the trade-off between privacy and performance. We conducted a thorough theoretical analysis of {\tt M$^2$FDP}, identifying conditions under which the algorithm will converge sublinearly to a controllable region around a stationary point, and revealing the impact of different system factors on the privacy-utility trade-off. Based on our analysis, we developed an adaptive control algorithm to jointly optimize communication energy, latency, and the stationarity gap while enforcing the sub-linear rate and meeting desired privacy requirements. Numerical evaluations confirmed {\tt M$^2$FDP}'s superior training performance and improvements in resource efficiency compared to existing DP-infused FL/HFL algorithms.      

Future works will focus on developing more generalized privacy-preserving algorithms for MFL, capable of withstanding diverse types of attacks. Additionally, we aim to investigate the impact of network design on the effectiveness of various defense strategies across different attack frameworks.

\bibliographystyle{IEEEtran}
\bibliography{sample-base}
 

\newpage

\begingroup
\onecolumn
\raggedbottom

\setcounter{lemma}{0}
\setcounter{theorem}{0}
\setcounter{assumption}{0}
\setcounter{corollary}{0}

\appendix

\subsection{Preliminaries and Outline}\label{app:notations}
\noindent To facilitate the proofs, we define the maximum $L2$-norm sensitivity from each layer's aggregation as:
\begin{equation}
    \Delta_{l} \overset{\Delta}{=} \max_{c_l \in \mathcal{S}_l} \Delta_{l, c_l}, \quad \forall l \in [1, L].
\end{equation}
Additionally, we define the node $d_{l,j}$ to be the parent node of a leaf node $j \in \mathcal{S}_L$ at a specific layer $l$ (which is a unique node in a hierarchical structure).

We will find the following result useful in the proofs:
\begin{fact}\label{fact:1} 
Consider $n$ random real-valued vectors $\mathbf x_1,\cdots,\mathbf x_n\in\mathbb R^m$, the following inequality holds: 
 \begin{equation}
     \sqrt{\mathbb E\left[\Big\Vert\sum\limits_{i=1}^{n} \mathbf x_i\Big\Vert^2\right]}\leq \sum\limits_{i=1}^{n} \sqrt{\mathbb E[\Vert\mathbf x_i\Vert^2]}.
 \end{equation}
\end{fact}
\begin{proof} Note that
    \begin{align}
        &\sqrt{\mathbb E\left[\Big\Vert\sum\limits_{i=1}^{n}\mathbf x_i\Big\Vert^2\right]}
        =
        \sqrt{\sum\limits_{i,j=1}^{n}\mathbb E [\mathbf x_i^\top\mathbf x_j]}
        \overset{(a)}{\leq}
\sum\limits_{i,j=1}^{n}\sqrt{\mathbb E [\Vert\mathbf x_i\Vert^2] \mathbb E[\Vert\mathbf x_j\Vert^2]]}
        =
        \sum\limits_{i=1}^{n} \sqrt{\mathbb E[\Vert\mathbf x_i\Vert^2]},
    \end{align}
    where $(a)$ follows from Holder's inequality, $\mathbb E[|XY|] \leq \sqrt{\mathbb E[|X|^2]\mathbb E[ |Y|^2]}$.
\end{proof}

The proofs of the main results, Theorem~\ref{thm:noncvx} and Corollary~\ref{cor:noncvx}, are given in Appendix~\ref{app:thm2}. The proofs of supporting lemmas referenced in these results are given in Appendix~\ref{app:lemmas}.

\subsection{Proofs of Theorem~\ref{thm:noncvx} and Corollary~\ref{cor:noncvx}} \label{app:thm2}
\begin{theorem} \label{thm:noncvx} 
        Under Assumptions~\ref{assump:SGD_noise} and~\ref{assump:genLoss}, upon using {\tt DP-HFL} for ML model training, if $\eta^t=\frac{\gamma}{\sqrt{t+1}}$ with $\gamma\leq\min\{\frac{1}{K^{\max}},\frac{1}{T}\}/\beta$, the cumulative average of global loss gradients satisfies
\begin{align*} 
        &\frac{1}{T}\sum_{t=1}^T \mathbb{E}\left\|\nabla F(w^{(t)})\right\|^2 \leq \frac{2\beta}{\sqrt{T+1}} \mathbb{E}\left[ F(w^{(1)}) -  F(w^{(T+1)})\right] + \frac{K^{\max}\left(G^2\left(1 + \frac{1}{\beta}\right) + \sigma^2\right)}{T}\\
        &+ \frac{8 L M (K^{\max})^4 q^2\log(1/\delta)}{\epsilon^2}\sum_{l=1}^{L}(1 - p_{l-1}^{\min})^2 \bigg(p_{l}^{\max}\frac{\alpha_l^2}{\prod_{l' = l}^{L-1} s_{l'}^2} + \sum_{m = l}^{L-1} \prod_{l' = l}^{m-1}(1 - p_{l'}^{\min})p_{m}^{\max}\frac{\alpha_m^2}{\prod_{l' = l}^{ m-1} s_{l'}\prod_{l'' = m}^{ L-1} s_{l''}^2} \\&+ \prod_{l' = l}^{ L-1}(1 - p_{l'}^{\min})\frac{\alpha_L^2}{\prod_{l' = l}^{ L-1} s_{l'}}\bigg)
\end{align*}
\end{theorem}

\begin{proof}
    Consider $t \in [1, T]$, for a given aggregation interval $K^t$, the global average of the local models follows the following dynamics:
    \begin{align}
    \label{global_w_iter}
        w^{(t+1)} =& \sum_{c = 1}^{N_1} \rho_c \left(w_{1,c}^{(t,K^t+1)} + n_{1,c}^{(t,K^t+1)}\right) \nonumber \\
        =& w^{(t)} - \eta^t \sum_{k=1}^{K^t} \sum_{c_1 = 1}^{N_1} \rho_{c_1} \sum_{c_2\in \mathcal{S}_{1,c_1}} \rho_{1,c_1,c_2}\ldots\sum_{j\in\mathcal{S}_{L-1,c_{L-1}}} \rho_{L-1, c_{L-1},j} {g}_j^{(t,k)}\\
        &+ \underbrace{\sum_{c=1}^{N_1}\rho_c n_{1,c}^{(t,K^t+1)} + \sum_{l=1}^{L-1} \sum_{\substack{k=1,\\ K_l^t | k }}^{K^t} \left(\sum_{c_1 = 1}^{N_1} \rho_{c_1} \sum_{c_2\in \mathcal{S}_{1,c_1}} \rho_{1,c_1,c_2}\ldots\sum_{c _l\in \mathcal{S}_l \cap \mathcal{N}_{U,l}} \sum_{i\in \mathcal{S}_{l,c_l}} \rho_{l,c_l,i}n_{l+1,i}^{(t,k)}\right)}_\text{$n_{DP}^t$}
    \end{align}
    On the other hand, the $\beta$-smoothness of the global function $F$ implies
    \begin{equation}
        F(w^{(t+1)}) \leq F(w^{(t)}) + \nabla F(w)^\top (w^{(t+1)} - w^{(t)}) + \frac{\beta}{2}\|w^{(t+1)} - w^{(t)}\|^2.
    \end{equation}
    By injecting this inequality into \eqref{global_w_iter}, taking the expectations on both sides of the inequality, and use the fact that $\mathbb{E}[n_{DP}^t] = 0$ yields:
    \begin{align}
    \label{eq:A_main_iter}
        &\mathbb{E}\left[ F(w^{(t+1)}) -  F(w^{(t)})\right] \leq \underbrace{- \eta^t \sum_{k=1}^{K^t} \mathbb{E}\left\langle \nabla F(w^{(t)}),  \sum_{c_1 = 1}^{N_1} \rho_{c_1} \sum_{c_2\in \mathcal{S}_{1,c_1}} \rho_{1,c_1,c_2}\ldots\sum_{j\in\mathcal{S}_{L-1,c_{L-1}}} \rho_{L-1, c_{L-1},j} \nabla F_j(w_j^{(t,k)})\right\rangle}_\text{(a)} \notag\\
        &+ \underbrace{\frac{(\eta^t)^2 \beta}{2} \mathbb{E}\left[\left\|\sum_{k=1}^{K^t} \sum_{c_1 = 1}^{N_1} \rho_{c_1} \sum_{c_2\in \mathcal{S}_{1,c_1}} \rho_{1,c_1,c_2}\ldots\sum_{j\in\mathcal{S}_{L-1,c_{L-1}}} \rho_{L-1, c_{L-1},j} \nabla F_j(w_j^{(t,k)})\right\|^2\right]}_\text{(b)} + \underbrace{\frac{\beta}{2}\mathbb{E}\left[\left\|n_{DP}^t\right\|^2\right]}_\text{(c)} + \frac{(\eta^tK^t)^2\beta \sigma^2}{2} 
    \end{align}
    To bound (a), we apply Lemma~\ref{lem:inner} (see Appendix \ref{app:lemmas}) to get
    \begin{align}
    \label{eq:A_bound_on_innerprod}
       &- \eta^t \sum_{k=1}^{K^t} \mathbb{E}\left\langle \nabla F(w^{(t)}),  \sum_{c_1 = 1}^{N_1} \rho_{c_1} \sum_{c_2\in \mathcal{S}_{1,c_1}} \rho_{1,c_1,c_2}\ldots\sum_{j\in\mathcal{S}_{L-1,c_{L-1}}} \rho_{L-1, c_{L-1},j} \nabla F_j(w_j^{(t,k)})\right\rangle
        \leq
        -\frac{\eta^t K^t}{2}\mathbb{E}\left\|\nabla F(w^{(t)})\right\|^2 \notag\\
        &-\frac{\eta^t}{2} \mathbb{E}\sum_{k=1}^{K^t}\left\|\sum_{c_1 = 1}^{N_1} \rho_{c_1} \sum_{c_2\in \mathcal{S}_{1,c_1}} \rho_{1,c_1,c_2}\ldots\sum_{j\in\mathcal{S}_{L-1,c_{L-1}}} \rho_{L-1, c_{L-1},j} \nabla F_j(w_j^{(t,k)})\right\|^2 + (\eta^t)^3(K^t)^3\beta^2G^2\notag\\
        &+\frac{2(\eta^t)^2 L\beta^2 M (K^t)^5 T q^2 G^2\log(1/\delta)}{\epsilon^2}\sum_{l=2}^{L}(1 - p_{l-1}^{\min})^2 \bigg(p_{l}^{\max}\frac{\alpha_l^2}{\prod_{l' = l}^{L-1} s_{l'}^2} \\&+ \sum_{m = l}^{L-1} \prod_{l' = l}^{m-1}(1 - p_{l'}^{\min})p_{m}^{\max}\frac{\alpha_m^2}{\prod_{l' = l}^{ m-1} s_{l'}\prod_{l'' = m}^{ L-1} s_{l''}^2} + \prod_{l' = l}^{ L-1}(1 - p_{l'}^{\min})\frac{\alpha_L^2}{\prod_{l' = l}^{ L-1} s_{l'}}\bigg)
    \end{align}
    To bound (b), we use Assumption~\ref{assump:genLoss} and get
    \begin{align}
    \label{eq:A_bound_on_grad}
        &\frac{(\eta^t)^2 \beta}{2} \mathbb{E}\left[\left\|\sum_{k=1}^{K^t} \sum_{c_1 = 1}^{N_1} \rho_{c_1} \sum_{c_2\in \mathcal{S}_{1,c_1}} \rho_{1,c_1,c_2}\ldots\sum_{j\in\mathcal{S}_{L-1,c_{L-1}}} \rho_{L-1, c_{L-1},j} \nabla F_j(w_j^{(t,k)})\right\|^2\right] \\
        &\leq \frac{(\eta^t)^2 \beta}{2} K^t\sum_{k=1}^{K^t} \sum_{c_1 = 1}^{N_1} \rho_{c_1} \sum_{c_2\in \mathcal{S}_{1,c_1}} \rho_{1,c_1,c_2}\ldots\sum_{j\in\mathcal{S}_{L-1,c_{L-1}}} \rho_{L-1, c_{L-1},j}\mathbb{E}\left[\left\| \nabla F_j(w_j^{(t,k)})\right\|^2\right]\\
        &\leq \frac{(\eta^tK^t)^2 \beta G^2}{2}
    \end{align}
    To bound (c), we apply Lemma~\ref{lem:bd_on_DP_noise} and get
    \begin{align}
    \label{eq:A_bound_on_DP}
        &\frac{\beta}{2}\mathbb{E}\left[\left\|n_{DP}^t\right\|^2\right] \\
        &\leq \frac{\beta}{2}\mathbb{E}\left[\left\| \sum_{c=1}^{N_1}\rho_c n_{1,c}^{(t,K^t+1)} + \sum_{l=1}^{L-1} (1 - p_{l}^{\min})\sum_{\substack{k=1,\\ K_l^t | k }}^{K^t} \left(\sum_{c_1 = 1}^{N_1} \rho_{c_1} \sum_{c_2\in \mathcal{S}_{1,c_1}} \rho_{1,c_1,c_2}\ldots\sum_{c _l\in \mathcal{S}_l \cap \mathcal{N}_{U,l}} \sum_{i\in \mathcal{S}_{l,c_l}} \rho_{l,c_l,i}n_{l+1,i}^{(t,k)}\right)\right\|^2\right]\\
        & \leq \frac{\beta L}{2}\mathbb{E}\left[\left\|\sum_{c=1}^{N_1}\rho_c n_{1,c}^{(t,K^t+1)}\right\|^2\right] \\
        &+ \frac{\beta L}{2}\sum_{l=1}^{L-1}(1 - p_{l}^{\min})^2\mathbb{E} \left[\left\| \sum_{\substack{k=1,\\ K_l^t | k }}^{K^t} \left(\sum_{c_1 = 1}^{N_1} \rho_{c_1} \sum_{c_2\in \mathcal{S}_{1,c_1}} \rho_{1,c_1,c_2}\ldots\sum_{c _l\in \mathcal{S}_l \cap \mathcal{N}_{U,l}} \sum_{i\in \mathcal{S}_{l,c_l}} \rho_{l,c_l,i}n_{l+1,i}^{(t,k)}\right)\right\|^2\right]\\
        &\leq \frac{ 2(\eta^t)^2 L\beta  M (K^t)^4 T q^2G^2\log(1/\delta)}{\epsilon^t}\sum_{l=1}^{L}(1 - p_{l-1}^{\min})^2 \bigg(p_{l}^{\max}\frac{\alpha_l^2}{\prod_{l' = l}^{L-1} s_{l'}^2} \\&+ \sum_{m = l}^{L-1} \prod_{l' = l}^{m-1}(1 - p_{l'}^{\min})p_{m}^{\max}\frac{\alpha_m^2}{\prod_{l' = l}^{ m-1} s_{l'}\prod_{l'' = m}^{ L-1} s_{l''}^2} + \prod_{l' = l}^{ L-1}(1 - p_{l'}^{\min})\frac{\alpha_L^2}{\prod_{l' = l}^{ L-1} s_{l'}}\bigg)
    \end{align}
    By substituting \eqref{eq:A_bound_on_innerprod}, \eqref{eq:A_bound_on_grad}, \eqref{eq:A_bound_on_DP} into \eqref{eq:A_main_iter}, and using Assumption~\ref{assump:genLoss} yields
    \begin{align}
        &\mathbb{E}\left[ F(w^{(t+1)}) -  F(w^{(t)})\right] \leq -\frac{\eta^t K^t}{2}\mathbb{E}\left\|\nabla F(w^{(t)})\right\|^2 + \frac{(\eta^2K^t)^2 \beta}{2} \left(G^2 + \sigma^2 \right)\\
        &+\frac{4(\eta^t)^3 L\beta^2 M (K^t)^5 T q^2 G^2\log(1/\delta)}{\epsilon^2}\sum_{l=2}^{L}(1 - p_{l-1}^{\min})^2 \bigg(p_{l}^{\max}\frac{\alpha_l^2}{\prod_{l' = l}^{L-1} s_{l'}^2} + \sum_{m = l}^{L-1} \prod_{l' = l}^{m-1}(1 - p_{l'}^{\min})p_{m}^{\max}\frac{\alpha_m^2}{\prod_{l' = l}^{ m-1} s_{l'}\prod_{l'' = m}^{ L-1} s_{l''}^2} \\&+ \prod_{l' = l}^{ L-1}(1 - p_{l'}^{\min})\frac{\alpha_L^2}{\prod_{l' = l}^{ L-1} s_{l'}}\bigg)\\
        &+\frac{ 2(\eta^t)^2 L\beta  M (K^t)^4 T q^2G^2\log(1/\delta)}{\epsilon^t}\sum_{l=1}^{L}(1 - p_{l-1}^{\min})^2 \bigg(p_{l}^{\max}\frac{\alpha_l^2}{\prod_{l' = l}^{L-1} s_{l'}^2} + \sum_{m = l}^{L-1} \prod_{l' = l}^{m-1}(1 - p_{l'}^{\min})p_{m}^{\max}\frac{\alpha_m^2}{\prod_{l' = l}^{ m-1} s_{l'}\prod_{l'' = m}^{ L-1} s_{l''}^2} \\&+ \prod_{l' = l}^{ L-1}(1 - p_{l'}^{\min})\frac{\alpha_L^2}{\prod_{l' = l}^{ L-1} s_{l'}}\bigg)
    \end{align}
   By selecting $\eta^t \leq \frac{1}{ K^t\beta}$ gives us
    \begin{align}
        &\frac{\eta^t K^t}{2}\mathbb{E}\left\|\nabla F(w^{(t)})\right\|^2 \leq \mathbb{E}\left[ F(w^{(t)}) -  F(w^{(t+1)})\right] + \frac{(\eta^2K^t)^2 \beta}{2} \left(G^2 + \sigma^2 \right)\\
        &+ (\eta^tK^2)^2\frac{6 L\beta  M (K^t)^2 T q^2G^2\log(1/\delta)}{\epsilon^2}\sum_{l=1}^{L}(1 - p_{l-1}^{\min})^2 \bigg(p_{l}^{\max}\frac{\alpha_l^2}{\prod_{l' = l}^{L-1} s_{l'}^2} + \sum_{m = l}^{L-1} \prod_{l' = l}^{m-1}(1 - p_{l'}^{\min})p_{m}^{\max}\frac{\alpha_m^2}{\prod_{l' = l}^{ m-1} s_{l'}\prod_{l'' = m}^{ L-1} s_{l''}^2} \\&+ \prod_{l' = l}^{ L-1}(1 - p_{l'}^{\min})\frac{\alpha_L^2}{\prod_{l' = l}^{ L-1} s_{l'}}\bigg)
    \end{align}
    Dividing both hand sides by $\frac{\eta^tK^t}{2}$ and plug in $K^{\max} = \max_t K^t$, averaging across global aggregation yields
    \begin{align}
        &\frac{1}{T}\sum_{t=1}^T \mathbb{E}\left\|\nabla F(w^{(t)})\right\|^2 \leq \frac{2\beta}{\sqrt{T+1}} \mathbb{E}\left[ F(w^{(1)}) -  F(w^{(T+1)})\right] + \frac{K^{\max}\left(G^2 + \sigma^2\right)}{T}\\
        &+ \frac{8 L M (K^{\max})^3 q^2\log(1/\delta)}{\epsilon^2}\sum_{l=1}^{L}(1 - p_{l-1}^{\min})^2 \bigg(p_{l}^{\max}\frac{\alpha_l^2}{\prod_{l' = l}^{L-1} s_{l'}^2} + \sum_{m = l}^{L-1} \prod_{l' = l}^{m-1}(1 - p_{l'}^{\min})p_{m}^{\max}\frac{\alpha_m^2}{\prod_{l' = l}^{ m-1} s_{l'}\prod_{l'' = m}^{ L-1} s_{l''}^2} \\&+ \prod_{l' = l}^{ L-1}(1 - p_{l'}^{\min})\frac{\alpha_L^2}{\prod_{l' = l}^{ L-1} s_{l'}}\bigg)
    \end{align}
\end{proof}

\begin{corollary} \label{cor:noncvx}
    Under Assumptions~\ref{assump:SGD_noise} and \ref{assump:genLoss}, if $\eta^t=\frac{\gamma}{\sqrt{t+1}}$ with $\gamma\leq\min\{\frac{1}{K^{\max}},\frac{1}{T}\}/\beta$, and let $m$ be the layer where all edge servers below it are all secure servers ($m$ is the lowest layer to have insecure servers), i.e.
    \begin{align}
    \label{eq:cor_condition}
        p_{l'}^{\max} = p_{l'}^{\min} = 1, \quad l' \in [m+1 , L],
    \end{align}
    then the cumulative average of global gradient satisfies
    \begin{align} \label{eq:cor_noncvx}
\small
        &\frac{1}{T}\sum_{t=1}^T \mathbb{E}\left\|\nabla F(w^{(t)})\right\|^2 \leq \frac{2\beta}{\sqrt{T+1}} \mathbb{E}\left[ F(w^{(1)}) -  F(w^{(T+1)})\right] \notag\\
        &+ \frac{K^{\max}\left(G^2\left(1 + \frac{1}{\beta}\right) + \sigma^2\right)}{T}\notag\\
        &+ \frac{8 L M (K^{\max})^4 q^2\log(1/\delta)}{\epsilon^2}\sum_{l=1}^{m}  \frac{(1 - p_{l-1}^{\min})^2p_l^{\max}\alpha_l^2}{\prod_{l'=l}^{L-1}s_{l'}^2}\notag\\
        &+ \frac{8 L M (K^{\max})^4 q^2\log(1/\delta)}{\epsilon^2}\sum_{l=1}^{m} \frac{(1 - p_{l-1}^{\min})^2(1 - p_l^{\min})(\alpha_l')^2}{\prod_{l'=l}^{m}s_{l'}\prod_{l''=m+1}^{L-1}s_{l''}^2}.
\end{align}
\end{corollary}
\begin{proof}
    Starting from the result from Theorem~\ref{thm:noncvx}, we can inject the condition from \eqref{eq:cor_condition}:
    \begin{align}
        &\sum_{l=1}^{L}(1 - p_{l-1}^{\min})^2 \bigg(p_{l}^{\max}\frac{\alpha_l^2}{\prod_{l' = l}^{L-1} s_{l'}^2} + \sum_{m' = l}^{L-1} \prod_{l' = l}^{m'-1}(1 - p_{l'}^{\min})p_{m'}^{\max}\frac{\alpha_{m'}^2}{\prod_{l' = l}^{ m'-1} s_{l'}\prod_{l'' = m'}^{ L-1} s_{l''}^2} + \prod_{l' = l}^{ L-1}(1 - p_{l'}^{\min})\frac{\alpha_L^2}{\prod_{l' = l}^{ L-1} s_{l'}}\bigg)\notag\\
        & = \sum_{l=1}^{m}(1 - p_{l-1}^{\min})^2 \bigg(p_{l}^{\max}\frac{\alpha_l^2}{\prod_{l' = l}^{m} s_{l'}^2} + \sum_{m' = l}^{m} \prod_{l' = l}^{m'}(1 - p_{l'}^{\min})p_{m'+1}^{\max}\frac{\alpha_{m'}^2}{\prod_{l' = l}^{ m'} s_{l'}\prod_{l'' = m'+1}^{ L-1} s_{l''}^2}\bigg)\notag\\
        &\leq \sum_{l=1}^{m}(1 - p_{l-1}^{\min})^2 \bigg(p_{l}^{\max}\frac{\alpha_l^2}{\prod_{l' = l}^{m} s_{l'}^2} + \sum_{m' = l}^{m} \prod_{l' = l}^{m'}(1 - p_{l'}^{\min})p_{m'+1}^{\max}\frac{\alpha_m^2}{\prod_{l' = l}^{ m} s_{l'}\prod_{l'' = m+1}^{ L-1} s_{l''}^2}\bigg) \notag\\
        &\leq \sum_{l=1}^{m}(1 - p_{l-1}^{\min})^2 \bigg(p_{l}^{\max}\frac{\alpha_l^2}{\prod_{l' = l}^{m} s_{l'}^2} + (1-p_{l}^{\min})\frac{(\alpha_l')^2}{\prod_{l' = l}^{ m} s_{l'}\prod_{l'' = m+1}^{ L-1} s_{l''}^2}\bigg)
    \end{align}
    Plugging this inequality back to Theorem~\ref{thm:noncvx} yields the result.
\end{proof}

\subsection{Proofs of Lemmas}\label{app:lemmas}
\begin{lemma}\label{lem:Delta} 
    Under Assumption~\ref{assump:genLoss}, the $L_2$-norm sensitivity of the exchanged gradients during local aggregations can be obtained as follows:
\begin{itemize}
    \item For any $l \in [1, L-1]$, and any given node $c_l \in \mathcal{S}_l$, the sensitivity is bounded by
    \begin{align}\label{eq:locAgg_egd_sens1}
        \Delta_{l, c_l} &= \max_{\mathcal{D}, \mathcal{D}'}\left\|\eta^t\sum_{k=1}^{K^t}\sum_{c_{l+1}\in \mathcal{S}_{l,c_l}}\rho_{l, c_l, c_{l+1}} \ldots \sum_{j\in \mathcal{S}_{L-1, c_{L-1}}} \rho_{L-1, c_{L-1}, j} \left(g_j^{(t,k)}(\mathcal{D}) - g_j^{(t,k)}(\mathcal{D}')\right)\right\| \\
        &\leq \frac{2\eta^t K^t G}{\prod_{l'=l}^{L-1} s_{l'}}, \quad \forall l \in [1, L-1]
    \end{align}
    \item For any given node $j \in \mathcal{S}_L$, the sensitivity is bounded by
    \begin{align}\label{eq:locAgg_egd_sens2}
        \Delta_{L,j} = \max_{\mathcal{D}, \mathcal{D}'}\left\|\eta^t\sum_{k=1}^{K^t}\left(g_j^{(t,k)}(\mathcal{D}) - g_j^{(t,k)}(\mathcal{D}')\right)\right\| \leq 2\eta^t K^t G
    \end{align}
\end{itemize}

\end{lemma}

    \begin{proof}
    (Case $l \in [1, L-1]$) From the condition that $\mathcal{D}$ and $\mathcal{D}'$ are adjacent dataset, which means the data samples only differs in one entry. As a result, all the summations within the norm, other than the sum for $k$, can be removed directly after upper bounding $\rho_{l, c_l, c_{l+1}}$. From previous definitions, we have:
    \begin{equation}
        \rho_{l, c_l, c_{l+1}} \leq \frac{1}{s_l}
    \end{equation}
As as result, we can bound $\Delta_{l, c_l}$ as
\begin{align}
    \Delta_{l, c_l} &\leq \Delta_{l}\\
    &\leq \frac{\eta^t}{\prod_{l'=l}^{L-1} s_{l'}} \sum_{k=1}^{K^t}\max_{\mathcal{D}, \mathcal{D}'}\left\|g_j^{(t,k)}(\mathcal{D}) - g_j^{(t,k)}(\mathcal{D}')\right\|\\
    &\leq \frac{\eta^t}{\prod_{l'=l}^{L-1} s_{l'}} \sum_{k=1}^{K^t}\max_{\mathcal{D}, \mathcal{D}'}\left\|g_j^{(t,k)}(\mathcal{D}) \right\|+\left\| g_j^{(t,k)}(\mathcal{D}')\right\|\\
    &\leq \frac{2\eta^t K^t G}{\prod_{l'=l}^{L-1} s_{l'}}
\end{align}
giving us the result in~\eqref{eq:locAgg_egd_sens1}.
    
(Case $l = L$) Similarly, for the leaf nodes, we can bound $\Delta_{L, j}$ as
        \begin{align}
    \Delta_{L, j} &\leq \Delta_{L}\\
    &\leq \eta^t \sum_{k=1}^{K^t}\max_{\mathcal{D}, \mathcal{D}'}\left\|g_j^{(t,k)}(\mathcal{D}) - g_j^{(t,k)}(\mathcal{D}')\right\|\\
    &\leq \eta^t \sum_{k=1}^{K^t}\max_{\mathcal{D}, \mathcal{D}'}\left\|g_j^{(t,k)}(\mathcal{D}) \right\|+\left\| g_j^{(t,k)}(\mathcal{D}')\right\|\\
    &\leq 2\eta^t K^t G
\end{align} giving us the result in~\eqref{eq:locAgg_egd_sens2}.
\end{proof}

\begin{lemma}\label{lem:bd_on_DP_noise}
For any given $k \in [1, K^t]$, and any layer $l \in [1, L-1]$, the DP noise that is passed up the $l^{\textrm{th}}$ layer from lower layers at local iteration $k$ can be bounded by:
\begin{align}
   &\mathbb{E}\left\|\sum_{i\in \mathcal{S}_{l,d_{l,i}}} \rho_{l,d_{l,i},i}n_{l+1,i}^{(t,k)} \right\|^2 \notag \\
   &\leq \frac{2(\eta^t)^2 M (K^t)^2 T q^2G^2\log(1/\delta)}{\epsilon^2} \left(p_{l+1}^{\max}\frac{\alpha_l^2}{\prod_{l' = l+1}^{L-1} s_{l'}^2} + \sum_{m = l+1}^{L-1} \prod_{l' = l+1}^{m-1}(1 - p_{l'}^{\min})p_{m}^{\max}\frac{\alpha_m^2}{\prod_{l' = l+1}^{ m-1} s_{l'}\prod_{l'' = m}^{ L-1} s_{l''}^2} \right)\\
   &+ \frac{2(\eta^t)^2 M (K^t)^2 T q^2G^2\log(1/\delta)}{\epsilon^2}\left(\prod_{l' = l+1}^{ L-1}(1 - p_{l'}^{\min})\frac{\alpha_L^2}{\prod_{l' = l+1}^{ L-1} s_{l'}}\right)
\end{align}
\end{lemma}
\begin{proof}
    We first move the norm into the sum
    \begin{align}
    \label{eq:sum_of_dp_noise}
        \mathbb{E}\left\|\sum_{i\in \mathcal{S}_{l,d_{l,i}}} \rho_{l,d_{i,j},i}n_{l+1,i}^{(t,k)} \right\|^2 \leq \sum_{i\in \mathcal{S}_{l,d_{l,i}}} \rho_{l,d_{i,j},i} \mathbb{E}\left\|n_{l+1,i}^{(t,k)}\right\|^2 
    \end{align}
    Since each $\mathbb{E}\left\|n_{l+1,i}^{(t,k)}\right\|^2$ is the variance of a guassian noise that's a linear combination of multiple noises, we can use Lemma~\ref{lem:Delta} to show that for any given $i$, by applying the total law of expectation:
    \begin{align}
        &\mathbb{E}\left\|n_{l+1,i}^{(t,k)}\right\|^2  \\
        &=  \mathbb E\Big[\Big\Vert{ n}_{l+1,i}^{(t,k)}\Big\Vert^2\Big|{ n}_{l+1,i}^{(t,k)}\sim \mathcal{N}(0, ({\sigma}_{l+1,i}^{(t,k)})^2)\Big]
    \\
    &= M\left({\sigma}_{l+1,i}^{(t,k)}\right)^2\\
    &\leq p_{l+1}^{\max} M\sigma^2 (\Delta_{l+1}) + (1 - p_{l+1}^{\min})p_{l+2}^{\max} \frac{M\sigma^2 (\Delta_{l+2})}{s_{l+1}} \\
    &+ (1 - p_{l+1}^{\min})(1 - p_{l+2}^{\min})p_{l+3}^{\max} \frac{M\sigma^2 (\Delta_{l+3})}{s_{l+1}s_{l+2}} + \ldots + \prod_{l' = l+1}^{ L-1}(1-p_{l'}^{\min}) \frac{M\sigma^2(\Delta_{L})}{\prod_{l' = l+1}^{ L-1} s_{l'}}\\
    &\leq \frac{2(\eta^t)^2 M (K^t)^2 T q^2G^2\log(1/\delta)}{\epsilon^2} \left(p_{l+1}^{\max}\frac{\alpha_l^2}{\prod_{l' = l+1}^{L-1} s_{l'}^2} + \sum_{m = l+1}^{L-1} \prod_{l' = l+1}^{m-1}(1 - p_{l'}^{\min})p_{m}^{\max}\frac{\alpha_m^2}{\prod_{l' = l+1}^{ m-1} s_{l'}\prod_{l'' = m}^{ L-1} s_{l''}^2} \right)\\
    &+ \frac{2(\eta^t)^2 M (K^t)^2 T q^2G^2\log(1/\delta)}{\epsilon^2}\left(\prod_{l' = l+1}^{ L-1}(1 - p_{l'}^{\min})\frac{\alpha_L^2}{\prod_{l' = l+1}^{ L-1} s_{l'}}\right)
    \end{align}
    We can plug this inequality back to \eqref{eq:sum_of_dp_noise} and obtain the results.
\end{proof}

\begin{lemma} \label{lem:inner} 
Under Assumption~\ref{assump:genLoss}, we have
    \begin{align*}
       &- \eta^t \sum_{k=1}^{K^t} \mathbb{E}\left\langle \nabla F(w^{(t)}),  \sum_{c_1 = 1}^{N_1} \rho_{c_1} \sum_{c_2\in \mathcal{S}_{1,c_1}} \rho_{1,c_1,c_2}\ldots\sum_{j\in\mathcal{S}_{L-1,c_{L-1}}} \rho_{L-1, c_{L-1},j} \nabla F_j(w_j^{(t,k)})\right\rangle
        \leq
        -\frac{\eta^t K^t}{2}\mathbb{E}\left\|\nabla F(w^{(t)})\right\|^2 \\
        &-\frac{\eta^t}{2} \sum_{k=1}^{K^t}\left\|\sum_{c_1 = 1}^{N_1} \rho_{c_1} \sum_{c_2\in \mathcal{S}_{1,c_1}} \rho_{1,c_1,c_2}\ldots\sum_{j\in\mathcal{S}_{L-1,c_{L-1}}} \rho_{L-1, c_{L-1},j} \nabla F_j(w_j^{(t,k)})\right\|^2 + (\eta^t)^3(K^t)^3\beta^2G^2\\
        &+\frac{(\eta^t)^2 L\beta^2 M (K^t)^5 T q^2 G^2\log(1/\delta)}{\epsilon^2}\sum_{l=2}^{L}(1 - p_{l-1}^{\min})^2 \bigg(p_{l}^{\max}\frac{\alpha_l^2}{\prod_{l' = l}^{L-1} s_{l'}^2} \\
        &+ \sum_{m = l}^{L-1} \prod_{l' = l}^{m-1}(1 - p_{l'}^{\min})p_{m}^{\max}\frac{\alpha_m^2}{\prod_{l' = l}^{ m-1} s_{l'}\prod_{l'' = m}^{ L-1} s_{l''}^2} + \prod_{l' = l}^{ L-1}(1 - p_{l'}^{\min})\frac{\alpha_L^2}{\prod_{l' = l}^{ L-1} s_{l'}}\bigg)
    \end{align*}
\end{lemma}

\begin{proof}  Since $-2 \mathbf a^\top \mathbf b = -\Vert \mathbf a\Vert^2-\Vert \mathbf b \Vert^2 + \Vert \mathbf a- \mathbf b\Vert^2$ holds for any two vectors $\mathbf a$ and $\mathbf b$ with real elements, we have

\begin{align}
    &- \eta^t \sum_{k=1}^{K^t} \mathbb{E}\left\langle \nabla F(w^{(t)}),  \sum_{c_1 = 1}^{N_1} \rho_{c_1} \sum_{c_2\in \mathcal{S}_{1,c_1}} \rho_{1,c_1,c_2}\ldots\sum_{j\in\mathcal{S}_{L-1,c_{L-1}}} \rho_{L-1, c_{L-1},j} \nabla F_j(w_j^{(t,k)})\right\rangle\\
    &= \frac{\eta^t}{2} \sum_{k=1}^{K^t} \Bigg[ -\left\|\nabla F(w^{(t)})\right\|^2 - \left\|\sum_{c_1 = 1}^{N_1} \rho_{c_1} \sum_{c_2\in \mathcal{S}_{1,c_1}} \rho_{1,c_1,c_2}\ldots\sum_{j\in\mathcal{S}_{L-1,c_{L-1}}} \rho_{L-1, c_{L-1},j} \nabla F_j(w_j^{(t,k)})\right\|^2 \\
    &+ \underbrace{\left\|\nabla F(w^{(t)})-\sum_{c_1 = 1}^{N_1} \rho_{c_1} \sum_{c_2\in \mathcal{S}_{1,c_1}} \rho_{1,c_1,c_2}\ldots\sum_{j\in\mathcal{S}_{L-1,c_{L-1}}} \rho_{L-1, c_{L-1},j} \nabla F_j(w_j^{(t,k)})\right\|^2}_\text{(a)}\Bigg]
\end{align}
Applying Assumption~\ref{assump:genLoss}, we can further bound (a) above as:
\begin{align}
    &\left\|\nabla F(w^{(t)})-\sum_{c_1 = 1}^{N_1} \rho_{c_1} \sum_{c_2\in \mathcal{S}_{1,c_1}} \rho_{1,c_1,c_2}\ldots\sum_{j\in\mathcal{S}_{L-1,c_{L-1}}} \rho_{L-1, c_{L-1},j} \nabla F_j(w_j^{(t,k)})\right\|^2\\
    \leq& \sum_{c_1 = 1}^{N_1} \rho_{c_1} \sum_{c_2\in \mathcal{S}_{1,c_1}} \rho_{1,c_1,c_2}\ldots\sum_{j\in\mathcal{S}_{L-1,c_{L-1}}} \rho_{L-1, c_{L-1},j} \left\| \nabla F(w^{(t)}) - \nabla F(w_j^{(t,k)}) \right\|^2\\
    \leq& \beta^2\sum_{c_1 = 1}^{N_1} \rho_{c_1} \sum_{c_2\in \mathcal{S}_{1,c_1}} \rho_{1,c_1,c_2}\ldots\sum_{j\in\mathcal{S}_{L-1,c_{L-1}}} \rho_{L-1, c_{L-1},j} \left\| w^{(t)} - w_j^{(t,k)} \right\|^2
\end{align}
Now for each layer $l$, we use the notation $d_{l,j}$ to denote the parent node of a leaf node $j \in \mathcal{S}_L$:
\begin{align}
&\left\|\nabla F(w^{(t)})-\sum_{c_1 = 1}^{N_1} \rho_{c_1} \sum_{c_2\in \mathcal{S}_{1,c_1}} \rho_{1,c_1,c_2}\ldots\sum_{j\in\mathcal{S}_{L-1,c_{L-1}}} \rho_{L-1, c_{L-1},j} \nabla F_j(w_j^{(t,k)})\right\|^2\\
    \leq& \beta^2\sum_{c_1 = 1}^{N_1} \rho_{c_1} \sum_{c_2\in \mathcal{S}_{1,c_1}} \rho_{1,c_1,c_2}\ldots\sum_{j\in\mathcal{S}_{L-1,c_{L-1}}} \rho_{L-1, c_{L-1},j} \mathbb{E}\left\| -\eta^t\sum_{k'=1}^{k-1} \widehat{g}_j^{(t,k')} + \sum_{l=1}^{L-1} (1 - p_l^{\min}) \sum_{\substack{k'\in \mathcal{K}_l^t }}  \sum_{i\in \mathcal{S}_{l,d_{l,j}}} \rho_{l,c,i}n_{l+1,i}^{(t,k')} \right\|^2\\
    \leq& 2\beta^2(\eta^t)^2  \sum_{c_1 = 1}^{N_1} \rho_{c_1} \sum_{c_2\in \mathcal{S}_{1,c_1}} \rho_{1,c_1,c_2}\ldots\sum_{j\in\mathcal{S}_{L-1,c_{L-1}}} \rho_{L-1, c_{L-1},j} \mathbb{E}\left\| \sum_{k'=1}^{k-1} \widehat{g}_j^{(t,k')} \right\|^2\\
    &+ 2L\beta^2 \sum_{c_1 = 1}^{N_1} \rho_{c_1} \sum_{c_2\in \mathcal{S}_{1,c_1}} \rho_{1,c_1,c_2}\ldots\sum_{j\in\mathcal{S}_{L-1,c_{L-1}}} \rho_{L-1, c_{L-1},j} \mathbb{E}\left\|\sum_{l=1}^{L-1} (1 - p_l^{\min})\sum_{\substack{k'\in \mathcal{K}_l^t }}  \sum_{i\in \mathcal{S}_{l,d_{l,j}}} \rho_{l,c,i}n_{l+1,i}^{(t,k')} \right\|^2
\end{align}

We can then use the Lemma~\ref{lem:bd_on_DP_noise}, and Assumption~\ref{assump:genLoss} to show that
\begin{align}
    &\mathbb{E}\left\|\nabla F(w^{(t)})-\sum_{c_1 = 1}^{N_1} \rho_{c_1} \sum_{c_2\in \mathcal{S}_{1,c_1}} \rho_{1,c_1,c_2}\ldots\sum_{j\in\mathcal{S}_{L-1,c_{L-1}}} \rho_{L-1, c_{L-1},j} \nabla F_j(w_j^{(t,k)})\right\|^2\\
    \leq & 2\beta^2(K^t)^2 G^2 \\&+  \frac{4  (\eta^t)^2 L\beta^2 M (K^t)^3 T q^2 G^2\log(1/\delta)}{\epsilon^2}\sum_{l=2}^{L} (1 - p_{l-1}^{\min})^2\bigg(p_{l}^{\max}\frac{\alpha_l^2}{\prod_{l' = l}^{L-1} s_{l'}^2} + \sum_{m = l}^{L-1} \prod_{l' = l}^{m-1}(1 - p_{l'}^{\min})p_{m}^{\max}\frac{\alpha_m^2}{\prod_{l' = l}^{ m-1} s_{l'}\prod_{l'' = m}^{ L-1} s_{l''}^2} \\
    &+ \prod_{l' = l}^{ L-1}(1 - p_{l'}^{\min})\frac{\alpha_L^2}{\prod_{l' = l}^{ L-1} s_{l'}}\bigg)
\end{align}

\end{proof}

\newpage
\subsection{Privacy Analysis for {\tt M$^2$FDP}} \label{app:privacy_analysis}
{\color{black}
We analyze the privacy guarantee for differential privacy on any insecure node $w_{l,c}^{(t,k)}$ at layer $l$. The privacy is guaranteed when adding noise $n_{l,c}^{(t,k)}$ when aggregating to parent nodes at layer $l-1$, as shown in \eqref{eq:insecure_local_aggr}. Hence we analyze the noise added to the model.
\begin{itemize}
    \item \textbf{Case 1: Node $c \in \mathcal{S}_{l}$ in layer $l$ is trustworthy.} In this case, a new noise $n_{l,c}^{(t,k)}$ needs to be generated and be added onto the $w_{l,c}^{(t,k)}$. Then, based on Proposition~\ref{prop:GM}, $(\epsilon,\delta)$-DP is directly guaranteed when $n_{l,c}^{(t,k)}$ is sampled from a gaussian distribution $\mathcal{N}(0, \sigma^2 (\Delta_{l,c})\mathbf{I}_M)$ with 
    \begin{equation}
        \sigma(\Delta_{l,c})^2 = \alpha_l^2 \frac{q^2\Delta_{l,c}^2L\log(1/\delta)}{\epsilon^2}
    \end{equation}
    where $\Delta_{l,c}$ is set to the upper bound of $L_2$-sensitivity $\frac{\eta^tK^tG}{\Pi_{l'=l}^{L-1}s_{l'}}$.

    \item \textbf{Case 2: Node $c \in \mathcal{S}_{l}$ in layer $l$ is not trustworthy, and all child nodes $i \in \mathcal{S}_{l,c}$ are trustworthy.} Since all child nodes are trustworthy, they will all generated a new noise $n_{l+1,i}^{(t,k)}$ when aggregating to node $c$. As a result, the noise added onto node $c$ will be a stochastic combination of all $n_{l+1,i}^{(t,k)}$ noises. 
    \begin{align}
        w_{l,c}^{(t,k)} &= \sum_{i \in \mathcal{S}_{l,c}}\rho_{l,c,i} (w_{l+1,i}^{(t,k)} +n_{l+1,i}^{(t,k)})\\
        &= \sum_{i \in \mathcal{S}_{l,c}}\rho_{l,c,i} w_{l+1,i}^{(t,k)} + \sum_{i \in \mathcal{S}_{l,c}}\rho_{l,c,i} n_{l+1,i}^{(t,k)}
    \end{align}
    Since this noise is not actually generated in the aggregation process but transmitted when model weights are combined and pass through the hierarchy, we define a auxiliary noise $\hat{n}_{l,c}^{(t,k)} = \sum_{i \in \mathcal{S}_{l,c}}\rho_{l,c,i}n_{l+1,i}^{(t,k)}$. Thus we can show that this noise can be viewed as a sampled noise from the Gaussian $\mathcal{N}(0, (\hat{\sigma}_{l,c}^{(t,k)})^2\mathbf{I}_M)$. Where:
    \begin{align}
        (\hat{\sigma}_{l,c}^{(t,k)})^2 &= \sum_{i \in \mathcal{S}_{l,c}}\rho_{l,c,i}^2 \sigma(\Delta_{l+1,i})^2\\
        &= \alpha_{l+1}^2 \frac{q^2L\log(1/\delta)}{\epsilon^2}\cdot \sum_{i \in \mathcal{S}_{l,c}}\rho_{l,c,i}^2 \Delta_{l+1,i}^2\\
        &= \alpha_{l+1}^2 \frac{q^2L\log(1/\delta)}{\epsilon^2}\cdot \sum_{i \in \mathcal{S}_{l,c}}\rho_{l,c,i}^2 (\frac{\eta^tK^t G}{\Pi_{l'=l+1}^{K-1}s_{l'}})^2\\
        &\geq \alpha_{l+1}^2 \frac{q^2L\log(1/\delta)}{\epsilon^2} \cdot \Delta_{l,c}^2
    \end{align}
    Hence by selecting e.g. $\alpha_l = \alpha_{l+1}$, Proposition~\ref{prop:GM} holds true and $(\epsilon,\delta)$-DP is guaranteed.

\item \textbf{Case 3: Node $c \in \mathcal{S}_{l}$ in layer $l$ is not trustworthy, and there are in total $s_{c,l'} \subseteq \mathcal{N}_{l'}$ nodes with noise generated and added during the aggregation process at layer $l'\in [l+1, L]$ from edge devices to node $c$.} In this case, the noise can be viewed as a sampled noise from the Gaussian $\mathcal{N}(0, (\hat{\sigma}_{l,c}^{(t,k)})^2\mathbf{I}_M)$. Then, there exists a set of stochastic weights $\{\phi_{l',i}\}_{i\in s_{c,l'},l' \in [l+1,L]}$ such that
    \begin{align}
        (\hat{\sigma}_{l,c}^{(t,k)})^2 
        &= (\sum_{l'=l+1}^L\sum_{i \in s_{c,l'}} \phi_{l',i}^2 \alpha_{l'}^2) \frac{q^2L\log(1/\delta)}{\epsilon^2}\cdot \sum_{l'=l+1}^L\sum_{i \in s_{c,l'}} \phi_{l',i}^2 \Delta_{l',i}^2
    \end{align}
    We can then define a new set of stochastic weights $\{\phi_{l'}\}_{l' \in [l+1,L]}$ and simply the equation:
    \begin{align}
         (\hat{\sigma}_{l,c}^{(t,k)})^2 
        &= (\sum_{l'=l+1}^L \phi_{l'}^2 \alpha_{l'}^2) \frac{q^2L\log(1/\delta)}{\epsilon^2}\cdot \sum_{l'=l+1}^L \phi_{l'}^2 (\frac{\eta^tK^t G}{\Pi_{l''=l'}^{K-1}s_{l''}})^2\\
        &\geq (\sum_{l'=l+1}^L \phi_{l'}^2 \alpha_{l'}^2 )\frac{q^2L\log(1/\delta)}{\epsilon^2}\cdot \Delta_{l,c}^2
    \end{align}

    Hence by selecting e.g. $\alpha_l = \sqrt{\sum_{l'=l+1}^L \phi_{l'}^2 \alpha_{l'}^2}$, Proposition~\ref{prop:GM} holds true and $(\epsilon,\delta)$-DP is guaranteed.
\end{itemize}

Since the three cases covers all scenarios for insecure data transmission during aggregation for arbitrary layer $l$ at any iteration $t,k$, we show that our algorithm ensures $(\epsilon,\delta)$-DP throughout the whole network for the whole training process.
}




\pagebreak

\end{document}